%% file: compatibility.tex
\documentclass[reprint,superscriptaddress]{revtex4-2}
\pdfoutput=1
\usepackage[percent]{overpic} 
\usepackage[british]{babel}
\usepackage{epsfig,amsmath,amssymb,amsthm,mathrsfs,dsfont,hyperref,ctable,url,commath,mathtools}
\usepackage{color}
\usepackage{bm}
\usepackage{units}
\usepackage{amscd}
\usepackage{thmtools}

\input{myQcircuit}

\newtheorem{lemma}{Lemma} \newtheorem{proposition}{Proposition}

\newtheorem{corollary}{Corollary} \newtheorem{theorem}{Theorem}

\theoremstyle{definition}
 \newtheorem{definition}{Definition}
\theoremstyle{definition}
 \newtheorem*{definition*}{Definition}
\theoremstyle{remark}

\def\>{\rangle}
\def\<{\langle} 
\def\trnsfrm#1{\mathcal #1}\newcommand
\test[1]{\mathsf{#1}}
   \def\rA{{\rm A}}\def\rB{{\rm
    B}}\def\rC{{\rm C}}\def\rD{{\rm D}}    \def\rI{{\rm I}}   \def\rX{{\rm X}} \def\rY{{\rm Y}}
\def\rZ{{\rm Z}}\def\rU{{\rm U}}\def\rV{{\rm V}}\def\rW{{\rm W}}
\def\rN{{\rm N}}\def\rW{{\rm W}}\def\rL{{\rm L}}
\def\tA{\trnsfrm A}\def\tB{\trnsfrm B}\def\tC{\trnsfrm
  C}\def\tD{\trnsfrm D} \def\tS{\trnsfrm S} \def\tI{\trnsfrm
  I}\def\tU{\trnsfrm U} \def\tV{\trnsfrm V}
\def\tW{\trnsfrm W}
\def\tZ{\trnsfrm Z} \def\tR{\trnsfrm R}\def\tP{\trnsfrm P }\def\tS{\trnsfrm S}

\def\Cntset{{\mathsf{Eff}}} 
\def\Stset{{\mathsf{St}}}

 \def\Trnset{{\mathsf{Transf}}}
 
\def\K#1{\left|#1\right)}  \def\B#1{\left(#1\right|}
\def\SC#1#2{\left(#1\right|\left.\!#2\right)}  

\def\Reals{{\mathbb R}}

\usepackage{enumitem}
\usepackage{array}
\newcolumntype{?}{!{\vrule width 1.2pt}}
\usepackage{soul}

\definecolor{blue1}{rgb}{0.03, 0.27, 0.49}

\usepackage{pifont}
\usepackage{xcolor}

\begin{document}

\title{Incompatibility of observables, channels and instruments in information theories}

\author{Giacomo Mauro \surname{D'Ariano}}

\email{dariano@unipv.it}

\affiliation{QUIT group,   Physics    Dept.,   Pavia
  University, and  INFN Sezione di Pavia,  via Bassi 6,
  27100 Pavia, Italy}

\author{Paolo \surname{Perinotti}}

\email{paolo.perinotti@unipv.it}

\affiliation{QUIT    group,   Physics    Dept.,   Pavia
  University, and  INFN Sezione di Pavia,  via Bassi 6,
  27100 Pavia, Italy}

\author{Alessandro \surname{Tosini}}

\email{alessandro.tosini@unipv.it}

\affiliation{QUIT group, Physics Dept., Pavia University, and INFN
  Sezione di Pavia, via Bassi 6, 27100 Pavia, Italy}

\begin{abstract}
Every theory of information, including classical and quantum, can be studied in the framework of operational 
probabilistic theories---where the notion of test generalizes that of quantum instrument, namely a collection of 
quantum operations summing to a channel, and simple rules are given for the composition of tests in parallel and in 
sequence. Here we study the notion of compatibility for tests of an operational probabilistic theory. Following the 
quantum literature, we first introduce the notion of strong compatibility, and then we illustrate its ultimate 
relaxation, that we deem weak compatibility. It is shown that the two notions coincide in the case of observation 
tests---which are the counterpart of quantum POVMs---while there exist weakly compatible channels that are not 
strongly compatible. We prove necessary and sufficient conditions for a theory to exhibit incompatible tests. We show 
that a theory admits of incompatible tests if and only if some information cannot be extracted without disturbance.
\end{abstract}

\maketitle

\section{Introduction}

The existence of incompatible observables, namely measurable quantities that cannot be accessed simultaneously, is at the heart of the departure between classical and quantum physics. Moreover, the presence of incompatible observables stands as a common ground behind several quantum hallmarks, as non-locality~\cite{LANDAU198754,PhysRevLett.103.230402,doi:10.1126/science.1192065,6280443,Khrennikov:2019vy} and contextuality~\cite{10.2307/24902153,PhysRevLett.48.291,PhysRevA.71.052108,PhysRevLett.112.040401,Acin:2015tx,PhysRevLett.119.120505,Dzhafarov:2017wq,PhysRevLett.119.050504}, and quantum no-go theorems, such as no-cloning~\cite{Wootters:1982wj} and no-information without disturbance~\cite{Busch2009}.

All the above effects related to quantum incompatibility have been studied in the details both within the mathematical 
framework of quantum theory and from a more fundamental perspective as features of an arbitrary information theory. 
The same fundamental analysis can be carried out for the primitive notion of incompatibility, a program already 
started within quantum logics~\cite{Lahti:1980tq}, and carried on in the last decade by Heinosaari and 
coauthors~\cite{Teiko-incompatibility-2016,Teiko-incompatibility-2017}. In the latter works, an operational definition 
of compatibility for \emph{quantum observables} and \emph{quantum channels} is given and thoroughly characterized. 
According to this approach two observables (channels) are compatible if there exists another observable (channel) that 
jointly implements them. 

In this paper we expand the operational notion of compatible processes for an arbitrary theory of information, and 
deepen its analysis; to this end we consider the framework of \emph{operational probabilistic theories}
~\cite{PhysRevA.81.062348,dariano_chiribella_perinotti_2017}. Moreover, we look at the compatibility of two arbitrary 
devices, what in quantum theory are called \emph{quantum instruments} or more generally \emph{tests} in an operational 
probabilistic theory. The compatibility of observables and channels, which are two kinds of tests, will be included in 
this analysis as special cases.

Given two devices, say A and B, two inequivalent operational notions of compatibility between them can be introduced. 
In the first case one one can say that A and B are compatible if there exists a third device, say C, that jointly 
implements both A and B. This is the spirit of references~\cite{Teiko-incompatibility-2016,Teiko-incompatibility-2017} 
on compatibility of quantum observables and channels. Alternatively, as proposed in this work, one can say that A and 
B are compatible if there exists a third device, say C, that after realizing A, leaves the possibility open of 
converting it to B via a post-processing, and viceversa there exists another device, say D, that first realizes B and 
then A via post-processing. We formalize the above intuitive ideas into two notions of compatibility, here denoted 
\emph{strong} and \emph{weak compatibility}, for tests of an operational probabilistic theory. As suggested by the 
names, compatible tests according to the first notion are also compatible according to the second one, while the 
converse is not true. While the two conditions coincide for observation tests, they indeed differ for general tests; 
e.g.~in quantum theory every pair of channels is weakly compatible, while there are pairs that are not strongly 
compatible.

The fundamental question, once compatibility is thoroughly defined, is to what other features of a theory is
it related, or in technical terms what are necessary or sufficient conditions for a theory to exhibit incompatible 
tests. Answering this question allows one to understand quantum incompatibility in relation to other operational 
aspects of the theory, rather than observing it as a sheer consequence of its mathematical formulation. Here we show 
that compatibility of tests is intimately related to the possibility of measuring any conceivable quantity without 
introducing disturbance on the measured system, and in turn this is equivalent to the request that all systems of the 
theory are classical.

After briefly reviewing the framework of operational probabilistic theories in Section~\ref{sec:II}, in Section~\ref{sec:III} we define strong and weak compatibility of tests. In Section~\ref{sec:IV} we prove necessary and sufficient conditions for a theory that does not contain any weakly incompatible tests, here denoted \emph{fully compatible theory}. Among those conditions a theory is fully compatible if and only if it satisfies \emph{full information without disturbance}, namely all the information encoded in the states of the theory can be read without disturbance. 

\section{Operational probabilistic theories}\label{sec:II}

The framework in which we introduce and study the notion of compatibile devices is that of {\em operational 
probabilistic theories} (OPT). Here we provide a brief review of the framework, with focus on those aspects that  
will be used in the remainder. For details and proofs, we refer
to~\cite{PhysRevA.81.062348,dariano_chiribella_perinotti_2017}
  
The essential building blocks of an operational theory are 
\emph{tests}, \emph{events}, and \emph{systems}. A test
$\test{A}=\{\tA_i\}_{i\in\rX}$ is the collection of events $\tA_i$, where $i$
labels the element of the outcome space $\rX$ of the test. Tests generalize the notion are the of quantum 
instruments, i.e.~collections of quantum operations that sum to a channel. An input and an output \emph{system} are 
associated to any test (event) with the following diagrammatic representation
\begin{equation*}
\Qcircuit @C=1em @R=.7em @! R { & \poloFantasmaCn \rA \qw &
  \gate{\test{A}} & \qw \poloFantasmaCn \rB &\qw}\;\ ,\;\ 
  \Qcircuit @C=1em @R=.7em @! R { & \poloFantasmaCn \rA \qw &
  \gate{\tA_i} & \qw \poloFantasmaCn \rB &\qw}\,.
\end{equation*}
One can compose tests in sequence and in parallel. \emph{Sequential composition:} 
given two tests $\test{A}=\{\tA_i\}_{i\in\rX}$ and $\test{B}=\{\tB_j\}_{j\in\rY}$ such that the input system of $\test{B}$ coincides with the output system of $\test{A}$ we can define their sequential composition as the test given by the collection of events
\begin{equation*}
  \Qcircuit @C=1em @R=.7em @! R { & \poloFantasmaCn \rA \qw &
    \gate{\tB_j\tA_i} & \qw \poloFantasmaCn \rC &\qw}\;=\; 
  \Qcircuit @C=1em @R=.7em @! R {& \poloFantasmaCn{\rA}\qw & \gate{
      \tA_{i}} & \poloFantasmaCn{\rB}\qw &
    \gate{\tB_{j}}&\poloFantasmaCn {\rC}\qw&\qw}\,,
\end{equation*}
for $i\in\rX$ and $j\in\rY$.  A singleton test is a test containing a single {\em deterministic} event.
For every system $\rA$ there exists a unique singleton test
$\{\tI_{\rA} \}$ such that $\tI_{\rB} \tA=\tA\tI_{\rA}=\tA$ for every
event $\tA$ with input $\rA$ and output $\rB$: we call $\tI_\rA$
{\em identity} of system $\rA$.  \emph{Parallel composition:} for every couple of systems $(\rA,\rB)$ we can form the
composite system $\rA\rB$, on which we can perform tests
$(\test{A}\otimes \test{B})_{\rX\times\rY}$ with events
$\tA_i\otimes\tB_j$ represented as follows
\begin{equation*}
  \begin{aligned}
    \Qcircuit @C=1em @R=.7em @! R {& \qw \poloFantasmaCn \rA &
      \multigate{1} {\tA_i\otimes \tB_j} & \qw
      \poloFantasmaCn \rC &\qw\\
      & \qw \poloFantasmaCn \rB & \ghost {\tA_i\otimes \tB_j} & \qw
      \poloFantasmaCn \rD &\qw}
  \end{aligned}\ =\ 
  \begin{aligned}
    \Qcircuit @C=1em @R=.7em @! R {& \qw \poloFantasmaCn \rA & \gate
      {\tA_i} & \qw \poloFantasmaCn \rC &\qw\\ & \qw \poloFantasmaCn
      \rB & \gate {\tB_j} & \qw \poloFantasmaCn \rD &\qw}
  \end{aligned}\,,
\end{equation*}
and satisfying the condition
$ ( \tC_{h} \otimes \tD_{k} ) ( \tA_{i} \otimes \tB_{j} ) = ( \tC_{h}
\tA_{i} ) \otimes ( \tD_{k}\tB_{j} ).  $
In quantum theory the parallel composition is the usual tensor product of
linear maps. However, for a general OPT, the parallel composition may
not coincide with a tensor product. We keep the notation $\otimes$ for an arbitrary parallel composition rule.
Intuitively speaking, the parallel composition of two tests represents their independent action on the two subsystems 
of a composite system.

There exists a {\em trivial system} $\rI$ such that $\rA\rI=\rI\rA=\rA$ for every system $\rA$. Tests with
input system $\rI$ and output $\rA$ are called {\em preparation tests}
of $\rA$, with events denoted as boxes without the input wire
$\Qcircuit @C=.5em @R=.5em { \prepareC{\rho} & \poloFantasmaCn{\rA}
  \qw & \qw }$
(or in formula as round kets $\K{\rho}_\rA$), while the tests with input system $\rA$ and output $\rI$ are
called {\em observation tests} of $\rA$, made of events with no output wire
\( \Qcircuit @C=.5em @R=.5em { & \poloFantasmaCn{\rA} \qw & \measureD{
    c} } \)
(in formula round bras $\B{ c}_\rA$). We will use the Greek letters for preparation events, 
e.g.~$\test{P}=\{\rho_i\}_{i\in\rX }$ and Latin letters for observation events, e.g.~$\test{C}=\{c_j\}_{j\in\rX }$.

In summary, an operational theory is a collection of systems and tests, closed under parallel and sequential 
composition, with a trivial system, and an identity test for every system of the theory.

Via sequential and parallel composition of diagrams one can build up \emph{circuits} that correspond 
to arbitrary tests. A circuit is \emph{closed} if it starts with a
  preparation test and ends with an observation test, namely its overall input and output systems are trivial. An
  \emph{operational probabilistic theory} (OPT) is an operational
  theory where any closed circuit of tests corresponds to a
  probability distribution for the joint test. Compound tests from the
  trivial system to itself are independent, both for sequential and
  parallel composition, namely their joint probability distribution is
  given by the product of the component joint probability
  distributions. For example the application of an observation event
  $c_i$ after the preparation event $\rho_j$ corresponds to the closed
  circuit $ \SC{ c_i }{ \rho_j}_\rA $ and denotes the probability of
  the outcome $(i,j)$ of the observation test $\test{c}$ after the
  preparation test $\test{\rho}$ of system $\rA$, i.e.
  \begin{equation*}
  \begin{aligned}
    \Qcircuit @C=.5em @R=.5em { \prepareC{\rho_j} &
      \poloFantasmaCn{\rA} \qw & \measureD{ c_i }
    }
  \end{aligned}\coloneqq
\Pr[i,j\,|\,
  \begin{aligned}
  \Qcircuit @C=.5em @R=.5em { \prepareC{\test{P}} & \poloFantasmaCn{\rA}
    \qw & \measureD{ \test{C}}} 
  \end{aligned}\,].
  \end{equation*}
  
Given a system $\rA$ of a probabilistic theory we can quotient the set
of preparation events of $\rA$ by the equivalence relation
$ \K{\rho}_\rA\sim\K{\sigma}_\rA \Leftrightarrow \SC{ c }{ \rho
}_\rA=\SC{ c }{ \sigma }_\rA$
for every observation event $c$. Similarly, we can quotient
observation events.  The equivalence classes of preparation events and
observation events of $\rA$ will be denoted by the same symbols as
their elements $\K{\rho}_\rA$ and $\B{c}_\rA$, respectively, and will
be called \emph{state} and \emph{effect} for system $\rA$. We will denote by $\Stset(\rA)$, $\Cntset(\rA)$
  the sets of states and effects of system $\rA$.  States and effects
  are real-valued functionals on each other, and can be
  naturally embedded in reciprocally dual real vector spaces,
  $\Stset_{\mathbb{R}}(\rA)$ and $\Cntset_{\mathbb{R}}(\rA)$, whose
  dimension is assumed to be finite. An event $\tA$ with input system $\rA$ and output system $\rB$ induces a
  linear map from $\Stset_{\mathbb R}(\rA\rC)$ to
  $\Stset_{\mathbb R}(\rB\rC)$ for each ancillary system $\rC$. The
  collection of all these maps (for each system $\rC$) is a {\em transformation} from
  $\rA$ to $\rB$, that represents the event $\tA$. In the following, the symbols $\tA$ and
  $ \Qcircuit @C=.5em @R=.5em { & \poloFantasmaCn{\rA} \qw &
    \gate{\tA} & \poloFantasmaCn{\rB} \qw &\qw} $
  will be used to represent the transformation corresponding to
  the event $\tA$. We will denote by $\Trnset(\rA,\rB)$, with linear 
  span $\Trnset_\mathbb{R}(\rA,\rB)$, the set of transformations from $\rA$ to $\rB$. A linear map
  $\tA\in\Trnset_{\mathbb R}(\rA,\rB)$ is {\em admissible} if it
  locally preserves the set of states $\Stset(\rA\rC)$, namely
  $\tA\otimes\tI_{\rC}(\Stset(\rA\rC)) \subseteq \Stset(\rB\rC)$, for
  every $\rC$. Given two maps $\tA,\tA^\prime\in\Trnset{(\rA,\rB)}$, one has
$\tA=\tA^\prime$, if and only if
\begin{equation*}
\begin{aligned}
  \Qcircuit @C=1em @R=.7em @! R {\multiprepareC{1}{\Psi}& \qw
    \poloFantasmaCn \rA &  \gate{\tA} & \qw \poloFantasmaCn \rB &\multimeasureD{1}{a} \\
    \pureghost\Psi &\qw& \qw \poloFantasmaCn \rC & \qw&\ghost{a}}
\end{aligned}
~=~
\begin{aligned}
\Qcircuit @C=1em @R=.7em @! R {\multiprepareC{1}{\Psi}& \qw
    \poloFantasmaCn \rA &  \gate{\tA^\prime} & \qw \poloFantasmaCn \rB &\multimeasureD{1}{a} \\
    \pureghost\Psi &\qw& \qw \poloFantasmaCn \rC & \qw&\ghost{a}}
\end{aligned}\ ,
\end{equation*}
for every $\rC$, every $\Psi\in\Stset{(\rA\rC)}$, and every
$a\in\Cntset(\rB\rC)$, namely they give the same probabilities within
every possible closed circuit. 

A deterministic transformation (i.e.~a transformation with conditional probability equal to 1 in every circuit) is 
also called a \emph{channel}. A transformation $\tA\in\Trnset(\rA\to\rB)$ is called \emph{reversible} if there 
exists a transformation $\tB\in\Trnset(\rB\to\rA)$ such that $\tB\tA=\tI_\rA$ and $\tA\tB=\tI_\rB$. In this case the 
inverse $\tB$ of $\tA$ is denoted as $\tA^{-1}$. A related notion is that of a \emph{left-reversible} transformation, 
that is a transformation $\tA\in\Trnset(\rA\to\rB)$ for which there exists $\tB\in\Trnset(\rB\to\rA)$
such that $\tB\tA=\tI_\rA$. In this case we say that $\tB$ is a left-inverse of $\tA$. In turn, $\tB$ is 
\emph{right-reversible}, and $\tA$ is a right-inverse of $\tB$. Notice that, unlike the inverse of a reversible 
transformation, the left-inverse as well as the right-inverse might not be unique.

  We finally introduce the notions of \emph{refinable} and \emph{atomic} event.

\begin{definition}[Refinement and coarse-graining] A {\em refinement} of an event 
$\tA\in\Trnset(\rA,\rB)$ is a collection of events $\{\tB_i\}_{i\in\rX}$ from $\rA$ to $\rB$, such that there exists a
test $\{\tB_i\}_{i\in\rY }$ with $X\subseteq Y$ and
$\tA=\sum_{i\in\rX}\tB_i$. Conversely, $\tA$ is called the {\em coarse-graining} of the events $\{\tB_i\}_{i\in\rX}$. 
Accordingly, a test $\test{D}\coloneqq\{\tD_l\}_{l\in\rW}$ is a {\em refinement} of the test $\test{C}\coloneqq\{\tC_k\}_{k\in\rZ}$ 
if there exists a partition of $\rW$ into disjoint sets $\rW_k$ such that $\tC_k=\sum_{l\in\rW_k}\tD_l$ for every $k\in\rZ$; conversely, we say that $\test C$ is a {\em coarse-graining} of $\test D$.
\end{definition}
For simplicity we will also refer to a  refinement of $\tA$  as $\tA=\sum_{i\in\rX}\tB_i$, without specifying the test including the events $\tB_i$. 
    
\begin{definition}[Refinable and atomic events]\label{def:atomic-events} 
  Given two events $\tA,\tB\in\Trnset(\rA,\rB)$ we say that $\tB$ {\em refines} $\tA$ if there exist a refinement
  $\{\tB_i\}_{i\in\rX}$ of $\tA$ such that $\tB\in\{\tB_i\}_{i\in\rX}$. A refinement
  $\{\tB_i\}_{i\in\rX}$ of $\tA$ is {\em trivial} if $\tB_i=\lambda_i \tA$,
  $\lambda_i\in [0,1]$, for every $i\in\rX$. An event $\tA$ is
  {\em atomic} if it admits only of trivial refinements. An event is refinable
  if it is not atomic.
\end{definition}
As in quantum theory, the word \emph{pure} is used as
synonym of atomic in the special case of states, and as usual a state that is not pure will be called \emph{mixed}.
 
A test where the same event, scaled with some probability, occurs more than once is redundant: indeed, it corresponds 
to a test where, upon occurrence of the outcome corresponding to the given event, a coin is flipped or some other 
random variable is sampled, and the outcome is supplemented with the value of the random variable. It is clear that 
the supplemental information in this case has nothing to do with the system undergoing the test. It is then often 
convenient to consider tests where redundancy is removed: 
\begin{definition}[Non redundant test]
We call a test $\test{A}\coloneqq\{\tA_i\}_{i\in\rX}$ {\em non redundant} when for every pair $i,j\in\rX$ one has 
$\tA_i\neq\lambda \tA_j$ for $\lambda>0$.
\end{definition}
From a redundant test one can always achieve a "maximal non redundant" one by taking the test made of coarse grainings
of all the sets of proportional elements.

\subsection{Causality and post-processing}

A notion at the basis of compatibility between tests will be that of \emph{post-processing}, here introduced for 
operational probabilistic theories in analogy to the post-processing of quantum 
instruments~\cite{PhysRevA.103.022615}. The post-processing of a test $\test A$, roughly speaking, corresponds to a 
second test that is run at the output of the test $\test A$. As one expects, the most general post-processing of a 
test exploits the possibility of using its outcomes for conditioning the choice of the post-processing maps. However, 
not every OPT offers the possibility to post-process tests in an arbitrary way. For this reason, throughout the paper 
we will restrict to causal theories~\cite{Perinotti2020cellularautomatain}, where using outcomes to condition the 
choice of a subsequent test always leads to admissible tests.

\begin{definition}[Strongly causal theories] A theory is {\em strongly causal} if for every test $\test{A}\coloneqq \{\tA_i\}_{i\in\rX}\subseteq\Trnset{(\rA,\rB)}$ and every  collection of tests $\test{P}^{(i)}\coloneqq \{\tP^{(i)}_j\}_{j\in\rY}\subseteq\Trnset{(\rB,\rC)}$, labelled by the outcomes $i\in\rX$ of test $\test{A}$, the test $\test{B}\coloneqq\{\tB_{i,j}\}_{(i,j)\in\rX\times\rY}\subseteq{\Trnset{(\rA,\rC)}}$, with events
\begin{align*} 
 \Qcircuit @C=0.7em @R=1em 
    {&
       \qw\poloFantasmaCn{\rA}& \gate{\tB_{i,j}}& \qw\poloFantasmaCn{\rC}&\qw
                                                                       }
=  
\begin{aligned}
    \Qcircuit @C=1em @R=.7em @! R { 
      & \qw\poloFantasmaCn{\rA}
      &\gate{{\tA}_{i}} &\qw\poloFantasmaCn{\rB}
      &\gate{{\tP_j^{(i)}}}  &\qw\poloFantasmaCn{\rC}
      &\qw}
  \end{aligned}\ ,
\end{align*}
is a test of the theory.
\end{definition}

We remark that, as suggested by the name, this notion of causality is strictly stronger than the uniqueness of the 
deterministic effect~\cite{PhysRevA.81.062348}, that is often used as a definition of causality, and here we will call 
{\em weak causality} for the sake of clarity.
\begin{proposition}[Unique deterministic effect~\cite{PhysRevA.81.062348}] In a strongly causal theory, for every 
system there exists a unique deterministic effect.
\end{proposition}

For the remainder we will assume strong causality, unless otherwise stated. Accordingly, in the following we will 
assume that there exists a unique way to discard a system, the counterpart of the trace operator in quantum theory. 
Therefore, given a test 
$\test{A}\coloneqq \{\tA_i\}_{i\in\rX}\subseteq\Trnset{(\rA,\rB)}$ there exists a unique observation-test associated 
to it, which is the test $\test{a}\coloneqq\{a_i\}_{i\in\rX}\subseteq\Cntset{(\rA)}$ with events $\B{a_i}=\B{e}\tA_i$, 
and $\B{e}$ the deterministic effect of system $\rB$. We remind that in a (weakly) causal theory a transformation 
$\tC\in\Trnset{(\rA,\rB)}$ is a channel if and only if $\B{e}\tC=\B{e}$. Moreover, one can always safely extend a 
strongly causal theory so that one can always assume that every state is proportional to a deterministic one, 
including in the theory every state 
$\rho\in\Stset{(\rA)}$ such that $\SC{e}{\rho}=1$~\cite{Perinotti2020cellularautomatain}.

In a (weakly) causal theory, the uniqueness of the deterministic effect clearly introduces an asymmetry between input 
and output of a test or a transformation. In this context it is then reasonable to define a reversible dilation 
of a test $\{\tA_i\}_{i\in\rX}\subseteq\Trnset(\rA\to\rB)$ as follows: we say that $\{\tA_i\}_{i\in\rX}$ has a 
\emph{left-reversible dilation} if there exist a system $\tilde\rB$, a left-reversible transformation 
$\tV\in\Trnset(\rA\to\rB\tilde\rB)$, and an observation-test 
$\{a_i\}_{i\in\rX}\subseteq\Cntset(\tilde\rB)$, such that for every $i\in\rX$ one has
\begin{align}
\begin{aligned}
  &   \nonumber\Qcircuit @C=0.7em @R=1em 
    {&
       \qw\poloFantasmaCn{\rA}& \gate{\tA_i}& \qw\poloFantasmaCn{\rB}&\qw
                                                                       }
\end{aligned}\; =\; 
\begin{aligned}
    \Qcircuit @C=1em @R=.7em @! R { & \qw\poloFantasmaCn{\rA}
      &\multigate{1}{\tV} &\qw\poloFantasmaCn{\rB}
      &\qw\\
      &&\pureghost{\tV} &\qw\poloFantasmaCn{\tilde\rB}
      &\measureD{a_i} }
  \end{aligned}\,.
\end{align}
In the special case of a singleton test $\{\tA\}$, the above definition becomes that of a {reversible dilation} 
of channel $\tA$. Spelling out the definition, we see that $\tA$ has a left-reversible dilation if 
there exist a system $\tilde\rB$ and a left-reversible transformation 
$\tV\in\Trnset(\rA\to\rB\tilde\rB)$, such that
\begin{align}
\begin{aligned}
  &   \nonumber\Qcircuit @C=0.7em @R=1em 
    {&
       \qw\poloFantasmaCn{\rA}& \gate{\tA}& \qw\poloFantasmaCn{\rB}&\qw
                                                                       }
\end{aligned}\; =\; 
\begin{aligned}
    \Qcircuit @C=1em @R=.7em @! R { & \qw\poloFantasmaCn{\rA}
      &\multigate{1}{\tV} &\qw\poloFantasmaCn{\rB}
      &\qw\\
      &&\pureghost{\tV} &\qw\poloFantasmaCn{\tilde\rB}
      &\measureD{e} }
  \end{aligned}\,.
\end{align}

A related, but stronger notion is the following: we say that $\{\tA_i\}_{i\in\rX}$ has a 
\emph{reversible dilation} if there exist systems $\tilde\rA$ and $\tilde\rB$, a reversible transformation 
$\tU\in\Trnset(\rA\tilde\rA\to\rB\tilde\rB)$, a state $\eta\in\Stset(\tilde\rA)$, and an observation-test 
$\{a_i\}_{i\in\rX}\subseteq\Cntset(\tilde\rB)$, such that for every $i\in\rX$ one has
\begin{align}
\begin{aligned}
  &   \nonumber\Qcircuit @C=0.7em @R=1em 
    {&
       \qw\poloFantasmaCn{\rA}& \gate{\tA_i}& \qw\poloFantasmaCn{\rB}&\qw
                                                                       }
\end{aligned}\; =\; 
\begin{aligned}
    \Qcircuit @C=1em @R=.7em @! R { & \qw\poloFantasmaCn{\rA}
      &\multigate{1}{\tU} &\qw\poloFantasmaCn{\rB}
      &\qw\\
      \prepareC{\eta}&\qw\poloFantasmaCn{\tilde\rA}&\ghost{\tU} &\qw\poloFantasmaCn{\tilde\rB}
      &\measureD{a_i} }
  \end{aligned}\,.
\end{align}
Notice that a reversible dilation $(\tilde\rA,\tilde\rB,\eta,\tU,\{a_i\}_{i\in\rX})$ of an instrument 
$\{\tA_i\}_{i\in\rX}$ straightforwardly allows one to define a left-reversible one 
$(\tilde\rB,\tV,\{a_i\}_{i\in\rX})$, just taking
\begin{align}
\begin{aligned}
    \Qcircuit @C=1em @R=.7em @! R { & \qw\poloFantasmaCn{\rA}
      &\multigate{1}{\tV} &\qw\poloFantasmaCn{\rB}
      &\qw\\
      &&\pureghost{\tV} &\qw\poloFantasmaCn{\tilde\rB}
      &\qw}
  \end{aligned}\; \coloneqq\; 
\begin{aligned}
    \Qcircuit @C=1em @R=.7em @! R { & \qw\poloFantasmaCn{\rA}
      &\multigate{1}{\tU} &\qw\poloFantasmaCn{\rB}
      &\qw\\
      \prepareC{\eta}&\qw\poloFantasmaCn{\tilde\rA}&\ghost{\tU} &\qw\poloFantasmaCn{\tilde\rB}
      &\qw }
  \end{aligned}\,,
\end{align}
whose right-inverse $\tilde\tV$ is given e.g.~by
\begin{align}
\begin{aligned}
    \Qcircuit @C=1em @R=.7em @! R { & \qw\poloFantasmaCn{\rB}
      &\multigate{1}{\tilde\tV} &\qw\poloFantasmaCn{\rA}&\qw\\
      &\qw\poloFantasmaCn{\tilde\rB}&\ghost{\tilde\tV} &
      &}
  \end{aligned}\; \coloneqq\; 
\begin{aligned}
    \Qcircuit @C=1em @R=.7em @! R { & \qw\poloFantasmaCn{\rB}
      &\multigate{1}{\tU^{-1}} &\qw\poloFantasmaCn{\rA}
      &\qw\\
      &\qw\poloFantasmaCn{\tilde\rB}&\ghost{\tU^{-1}} &\qw\poloFantasmaCn{\tilde\rA}
      &\measureD{e} }
  \end{aligned}\,.
\end{align}

As anticipated, we can now introduce the notion of test post-processing, that is generally a conditioned test, 
possibly followed by a coarse-graining and/or the discarding of an ancillary system.

\begin{definition}(Test post-processing)
The test $\test{C}\coloneqq\{\tC_k\}_{k\in\rZ}\subseteq\Trnset{(\rA,\rC)}$ is a {\em post-processing} of the test 
$\test{A}\coloneqq \{\tA_i\}_{i\in\rX}\subseteq\Trnset{(\rA,\rB)}$ if there exists a collection of tests 
$\test{P}^{(i)}\coloneqq \{\tP^{(i)}_j\}_{j\in\rY}\subseteq\Trnset{(\rB,\rC\rC^\prime)}$, $i\in\rX$ and a partition of $\rX\times\rY$ into disjoint sets $\rW_k$ such that
\begin{align}\label{eq:cond-test}
 \Qcircuit @C=0.7em @R=1em 
    {&
       \qw\poloFantasmaCn{\rA}& \gate{\tC_{k}}& \qw\poloFantasmaCn{\rC}&\qw
                                                                       }
\ =  \sum_{(i,j)\in\rW_k}
\begin{aligned}
    \Qcircuit @C=1em @R=.7em @! R { 
      & \qw\poloFantasmaCn{\rA}
      &\gate{{\tA}_{i}} 
      &\qw\poloFantasmaCn{\rB}
      &\multigate{1}{{\tP_j^{(i)}}} &\qw\poloFantasmaCn{\rC}&\qw\\
      & 
      &
      & 
      &\pureghost{{\tP_j^{(i)}}}&\qw\poloFantasmaCn{\rC^\prime}&\measureD{e} 
      }
  \end{aligned}\; .
\end{align}
\end{definition}
We stress that assuming only the framework of operational probabilistic theories, without further requirements, a post-processing of a test would not necessarily correspond to a test of the theory, while for strongly causal theories any post-processing is admissible, since the conditioned tests $\test{B}\coloneqq\{\tB_{i,j}\}_{(i,j)\in\rX\times\rY}\subseteq{\Trnset{(\rA,\rC\rC^\prime)}}$ in Eq.~\eqref{eq:cond-test} is a test of the theory. 
%
%
%
%
%

\section{Compatible tests}\label{sec:III}

We now introduce a first operational notion of compatibility of tests, that we deem strong compatibility, based on 
the idea that two tests are compatible if there exists a procedure, namely another test, for their joint realization. 
Strong compatibility extends the compatibility of quantum channels studied in Refs.~\cite{Teiko-incompatibility-2016,Teiko-incompatibility-2017} to tests of an arbitrary theory.

\begin{definition}[Strong compatibility]\label{def:strong}
 We say that the tests  $\test{A}\coloneqq \{\tA_i\}_{i\in\rX}\subseteq\Trnset{(\rA,\rB)}$ and $\test{B}\coloneqq\{\tB_j\}_{j\in\rY}\subseteq\Trnset{(\rA,\rC)}$ are {\em strongly compatible}, and we write $\test{A}\Join\test{B}$, if there exists a test $\test{C}\coloneqq \{\tC_{k}\}_{k\in\rZ}\subseteq\Trnset{(\rA,\rB\rC)}$ such that for every $i\in\rX$ and $j\in\rY$ one has
\begin{subequations}
\begin{align}\label{eq:strong-a}
& \Qcircuit @C=0.7em @R=1em 
    {&
       \qw\poloFantasmaCn{\rA}& \gate{\tA_i}& \qw\poloFantasmaCn{\rB}&\qw
                                                                       }
                                                                       \; =
  \sum_{k\in\rV_i} 
  \begin{aligned}
    \Qcircuit @C=1em @R=.7em @! R { 
    & \qw\poloFantasmaCn{\rA}
      &\multigate{1}{{\tC}_{k}} &\qw\poloFantasmaCn{\rB}&\qw
      \\
      & & \pureghost{{\tC}_{k}} &\qw\poloFantasmaCn{\rC}
      &\measureD{e} }
  \end{aligned},\\\label{eq:strong-b}
& \Qcircuit @C=0.7em @R=1em 
    {&
       \qw\poloFantasmaCn{\rA}& \gate{\tB_j}& \qw\poloFantasmaCn{\rC}&\qw
                                                                       }
                                                                       \; =
  \sum_{k\in\rW_j} 
  \begin{aligned}
    \Qcircuit @C=1em @R=.7em @! R { 
    & \qw\poloFantasmaCn{\rA}
      &\multigate{1}{{\tC}_{k}} &\qw\poloFantasmaCn{\rB}&\measureD{e} 
      \\
      & & \pureghost{{\tC}_{k}} &\qw\poloFantasmaCn{\rC}
      &\qw}
  \end{aligned}
\end{align}
\end{subequations}
for two partitions of $\rZ$ into disjoint subsets $\rV_i$, $i\in\rX$ and $\rW_j$, $j\in\rY$, respectively. 
\end{definition}
The following proposition shows that, without loss of generality, the test that jointly realizes two strongly 
compatible tests can be chosen as a double-outcome test, whose outcomes are pairs with elements corresponding 
respectively to the outcomes of the compatible tests.

\begin{proposition}[Strong compatibility]
Two tests $\test{A}\coloneqq \{\tA_i\}_{i\in\rX}\subseteq\Trnset{(\rA,\rB)}$ and $\test{B}\coloneqq\{\tB_j\}_{j\in\rY}\subseteq\Trnset{(\rA,\rC)}$ are strongly compatible if and only if there exists a test $\test{D}\coloneqq\{D_{i,j}\}_{(i,j)\in \rX\times\rY}\subseteq\Trnset{(\rA,\rB\rC)}$, such that 
\begin{subequations}
\begin{align}\label{eq:strong-a-bis}
& \Qcircuit @C=0.7em @R=1em 
    {&
       \qw\poloFantasmaCn{\rA}& \gate{\tA_i}& \qw\poloFantasmaCn{\rB}&\qw
                                                                       }
                                                                       \; =
  \sum_{j} 
  \begin{aligned}
    \Qcircuit @C=1em @R=.7em @! R { 
    & \qw\poloFantasmaCn{\rA}
      &\multigate{1}{{\tD}_{i,j}} &\qw\poloFantasmaCn{\rB}&\qw
      \\
      & & \pureghost{{\tD}_{i,j}} &\qw\poloFantasmaCn{\rC}
      &\measureD{e} }
  \end{aligned},\\\label{eq:strong-b-bis}
& \Qcircuit @C=0.7em @R=1em 
    {&
       \qw\poloFantasmaCn{\rA}& \gate{\tB_j}& \qw\poloFantasmaCn{\rC}&\qw
                                                                       }
                                                                       \; =
  \sum_{i} 
  \begin{aligned}
    \Qcircuit @C=1em @R=.7em @! R { 
    & \qw\poloFantasmaCn{\rA}
      &\multigate{1}{{\tD}_{i,j}} &\qw\poloFantasmaCn{\rB}&\measureD{e} 
      \\
      & & \pureghost{{\tD}_{i,j}} &\qw\poloFantasmaCn{\rC}
      &\qw}
  \end{aligned}\,.
\end{align}
\end{subequations}
\end{proposition}

\begin{proof}
Since $\test{A}$ and $\test{B}$ are strongly compatible then according to Definition~\ref{def:strong}, Eqs.~\eqref{eq:strong-a} and~\eqref{eq:strong-b} hold for a test $\test{C}\coloneqq \{\tC_{k}\}_{k\in\rZ}\subseteq\Trnset{(\rA,\rB\rC)}$. 
Consider now the disjoint partition of $\rZ$ given by $\cup_{i,j}\rU_{ij}=\rZ$, with $\rU_{ij}=\rV_i\cap\rW_j$, and define the test 
$\test{D}\coloneqq\{\tD_{i,j}\}_{(i,j)\in\rX\times\rY}\subseteq{(\rA,\rB\rC)}$, with
\begin{equation*}
\tD_{i,j}=\sum_{k\in\rV_i\cap\rW_j}\tC_k.
\end{equation*}
It is straightforward to see that $\test D$ satisfies Eqs.~\eqref{eq:strong-a-bis} and~\eqref{eq:strong-b-bis}.
The converse implication is trivial.
\end{proof}

Another operational notion of compatibility of tests, that we will show to be strictly weaker than the previous one, is proposed in the following. We say two tests $\test{A}$ and $\test{B}$ are compatible if there exists a realization of $\test{A}$ that can be post-processed to get $\test{B}$, and there exists a realization of $\test{B}$ that can be post-processed to get $\test{A}$. In other words it is possible to realize one of the two tests leaving open the possibility of deciding later to make the other one. 
  
\begin{definition}\label{def:weak}
We say that the test $\test{A}\coloneqq \{\tA_i\}_{i\in\rX}\subseteq\Trnset{(\rA,\rB)}$ {\em does not exclude} the 
test $\test{B}\coloneqq\{\tB_j\}_{j\in\rY}\subseteq\Trnset{(\rA,\rC)}$, and write $\test{A}\rightarrow\test{B}$, if 
one has
\begin{subequations}
 \begin{align}\label{eq:weak-a} 
& \Qcircuit @C=0.7em @R=1em 
    {&
       \qw\poloFantasmaCn{\rA}& \gate{\tA_i}& \qw\poloFantasmaCn{\rB}&\qw
                                                                       }
                                                                       \; =
  \sum_{k\in\rV_i} 
  \begin{aligned}
    \Qcircuit @C=1em @R=.7em @! R { 
    & \qw\poloFantasmaCn{\rA}
      &\multigate{1}{{\tC}_{k}} &\qw\poloFantasmaCn{\rB}&\qw
      \\
      & & \pureghost{{\tC}_{k}} &\qw\poloFantasmaCn{\rB'}
      &\measureD{e} }
  \end{aligned},\\\label{eq:weak-b}
& \Qcircuit @C=0.7em @R=1em 
    {&
       \qw\poloFantasmaCn{\rA}& \gate{\tB_j}& \qw\poloFantasmaCn{\rC}&\qw
                                                                       }
                                                                       \; =
  \sum_{(k,l)\in\rW_j} 
  \begin{aligned}
    \Qcircuit @C=1em @R=.7em @! R { 
      & \qw\poloFantasmaCn{\rA}
      &\multigate{1}{{\tC}_{k}} &\qw\poloFantasmaCn{\rB}
      &\multigate{1}{{\tP_l^{(k)}}}  &\qw\poloFantasmaCn{\rC}
      &\qw
      \\
      & 
      & \pureghost{{\tC}_{k}} &\qw\poloFantasmaCn{\rB'}
      &  \ghost{{\tP_l^{(k)}}} &\qw\poloFantasmaCn{\rC'}
      &\measureD{e}}
  \end{aligned}\,
\end{align}
\end{subequations}
for some test $\test{C}\coloneqq \{\tC_{k}\}_{k\in\rZ}\subseteq\Trnset{(\rA,\rB\rB')}$, post-processing tests 
$\test{P}^{(k)}\coloneqq \{\tP_l^{(k)}\}_{l\in\rL}\subseteq\Trnset{(\rB\rB',\rC\rC')}$, and disjoint partions $\rV_i, i\in\rX$ and $\rW_j,j\in\rY$ of $\rZ$ and $\rZ\times\rL$ respectively. 
\end{definition}

The above notion allows us now to define weak compatibility as follows.

\begin{definition}[Weak compatibility]\label{def:weak-fin}
We say that $\test{A}$ and $\test{B}$ are {\em weakly-compatible}, and write $\test{A}\leftrightarrow\test{B}$, if 
$\test A$ does not exclude $\test B$, and $\test B$ does not exclude $\test A$.
\end{definition}

Notice that the condition $\test{A}\rightarrow\test{B}$ does not imply $\test{B}\rightarrow\test{A}$. Moreover, given 
two weakly compatible tests $\test{A}\leftrightarrow\test{B}$, the test $\test{C}$ used for realizing $\test{A}$, and 
subsequently post-processing it to $\test{B}$ ($\test{A}\rightarrow\test{B}$), may be different from the test 
$\test{C}^\prime$ used to first realize $\test{B}$ and then  $\test{A}$ ($\test{B}\rightarrow\test{A}$). This is in 
the spirit of keeping the present notion of compatibility condition the weakest possible.

\subsection{Comparing strong and weak compatibility}

In this section we compare the two definitions of compatibility provided so far. First we prove that strong-compatibility always implies weak-compatibility.

\begin{proposition}\label{prop:stron-weak}  
Given two tests  $\test{A}\coloneqq \{\tA_i\}_{i\in\rX}\subseteq\Trnset{(\rA,\rB)}$ and $\test{B}\coloneqq\{\tB_j\}_{j\in\rY}\subseteq\Trnset{(\rA,\rC)}$, if they are strongly compatible then they are also weakly compatible.
\end{proposition}

\begin{proof}
By hypothesis $\test{A}\Join\test{B}$. According to Definition~\ref{def:strong} there exists a test $\test{C}\coloneqq\{C_{k}\}_{k\in\rZ}\subseteq\Trnset{(\rA,\rB\rC)}$ for the joint realization of $\test{A}$ and $\test{B}$ as in Eqs.~\eqref{eq:strong-a} and~\eqref{eq:strong-b}, respectively. Now we prove that it is also $\test{A}\leftrightarrow{B}$. Let us start with  $\test{A}\rightarrow{B}$ and show that Eqs.~\eqref{eq:weak-a} and~\eqref{eq:weak-b} in Definition~\ref{def:weak} can be satisfied. First one can rewrite Eqs.~\eqref{eq:strong-a} and~\eqref{eq:strong-b} as follows
\begin{subequations}\nonumber
 \begin{align} 
& \Qcircuit @C=0.7em @R=1em 
    {&
       \qw\poloFantasmaCn{\rA}& \gate{\tA_i}& \qw\poloFantasmaCn{\rB}&\qw
                                                                       }
                                                                       \; =
  \sum_{k\in\rV_i} 
  \begin{aligned}
    \Qcircuit @C=1em @R=.7em @! R { 
    & \qw\poloFantasmaCn{\rA}
      &\multigate{1}{{\tC}_{k}} &\qw\poloFantasmaCn{\rB}&\qw
      \\
      & & \pureghost{{\tC}_{k}} &\qw\poloFantasmaCn{\rC}
      &\measureD{e} }
  \end{aligned},\\
& \Qcircuit @C=0.7em @R=1em 
    {&
       \qw\poloFantasmaCn{\rA}& \gate{\tB_j}& \qw\poloFantasmaCn{\rC}&\qw
                                                                       }
                                                                       \; =
  \sum_{k\in\rW_j} 
  \begin{aligned}
    \Qcircuit @C=1em @R=.7em @! R { 
      & \qw\poloFantasmaCn{\rA}
      &\multigate{1}{{\tC}_{k}} &\qw\poloFantasmaCn{\rB}
      &\multigate{1}{{\tS}}  &\qw\poloFantasmaCn{\rC}
      &\qw
      \\
      & 
      & \pureghost{{\tC}_{k}} &\qw\poloFantasmaCn{\rC}
      &  \ghost{{\tS}} &\qw\poloFantasmaCn{\rB}
      &\measureD{e}}
  \end{aligned}\,
\end{align}
\end{subequations}
where $\tS\in\Trnset{(\rB\rC,\rB\rC)}$ is the deterministic transformation swapping systems $\rB$ and $\rC$. Then we observe that the two equations above also give instances of Eqs.~\eqref{eq:weak-a} and~\eqref{eq:weak-b}. The first equation coincides with Eq.~\eqref{eq:weak-a} upon identyfing system $\rB'$ with system $\rC$ and system $\rC'$ with system $\rB$. The second is a special case of Eq.~\eqref{eq:weak-b} where the post-processing tests $\test{P}^{(k)}$ are all chosen equal to the singleton test made of the swap gate $\tS$. Accordingly, one has $\rZ\times\rL=\rZ$, since $\rL$ is trivial, and the disjoint partition $\rW_j,j\in\rY$ of $\rZ$, is thus a disjoint partition of $\rZ\times \rL$ as well.
Analogously, one can verify that  $\test{B}\rightarrow\test{A}$ holds.
\end{proof}

We continue the comparison in the special cases of observation tests, say tests having classical outcomes, and for channels, say tests having non trivial output system but only one possible outcome.

\subsubsection{The case of observables}

Using the definitions presented in this section one can study compatibility between observation tests as a special case. In quantum theory observation tests are characterized by having a nontrivial input system, but a trivial output one, thus only producing a classical outcome. This case is different from instruments, where both the input and the output are non-trivial, and thus one has both a quantum output and a classical outcome. Also for an arbitrary theory  observation tests only produce classical outcomes. The scope of this section is to derive the consequences of  this simplification on the relation between strong and weak compatibility for observation tests.

The following proposition shows the collapse of the two notions of compatibility for two observation tests, and the 
equivalence of both conditions with the the possibility of detecting simultaneously the outcomes of the two tests.

\begin{proposition}[Compatible observation tests]\label{prop:compatible-observation-tests}
  For two observation tests $\test{A}\coloneqq \{a_i\}_{i\in\rX}\in\Cntset{(\rA)}$ and
  $\test{B}\coloneqq\{b_j\}_{j\in\rY}\in\Cntset{(\rA)}$ the following statements are equivalent
  \begin{enumerate}
  \item\label{obs:1} $\test{A}\Join \test{B}$, namely $\test{A}$ and $\test{B}$ are strongly-compatible
  \item\label{obs:2} $\test{A}\leftrightarrow \test{B}$, namely $\test{A}$ and $\test{B}$ are weakly-compatible
  \item\label{obs:3} $\test{A}\rightarrow \test{B}$ (or $\test{B}\rightarrow \test{A}$)
  \item\label{obs:4} There exists an oservation test $\test{C}\coloneqq\{c_{i,j}\}_{(i,j)\in\rX\times\rY}\subseteq\Cntset{(\rA)}$ such that 
\begin{equation}\label{eq:compatible-obs}
\begin{aligned}
&  \Qcircuit @C=0.7em @R=1em 
  {&
  \qw\poloFantasmaCn{\rA} &\measureD{a_i}
                            }
                            \; = \;
                            \sum_{j\in\rY}
\begin{aligned}
  \Qcircuit @C=1em @R=.7em @! R {&\qw\poloFantasmaCn{\rA}&\measureD{c_{i,j}} }
    \end{aligned}\,,
\\
&  \Qcircuit @C=0.7em @R=1em 
  {&
  \qw\poloFantasmaCn{\rA}& \measureD{b_j}
                                         }
                                         \; = \;
\sum_{i\in\rX}
\begin{aligned}
  \Qcircuit @C=1em @R=.7em @! R {& 
\qw\poloFantasmaCn{\rA}&\measureD{c_{i,j}}
  }
    \end{aligned}.
\end{aligned}
\end{equation}
\end{enumerate}
\end{proposition}

\begin{proof}
We show the equivalence proving the chain of implications $\ref{obs:1}\Rightarrow\ref{obs:2}\Rightarrow\ref{obs:3}\Rightarrow\ref{obs:4}\Rightarrow\ref{obs:1}$. 
The implication $\ref{obs:1}\Rightarrow\ref{obs:2}$ follows from Proposition~\ref{prop:stron-weak}. 
The implication $\ref{obs:2}\Rightarrow\ref{obs:3}$ is trivial, since $\test{A}\leftrightarrow\test{B}$ if and only if it is both $\test{A}\rightarrow\test{B}$ and $\test{B}\rightarrow\test{A}$.
Consider now the implication $\ref{obs:3}\Rightarrow\ref{obs:4}$. If $\test{A}\rightarrow\test{B}$ then 
by Definition~\ref{def:weak} there exist a test $\test{D}\coloneqq\{\tD_k\}_{k\in\rZ}\subseteq\Trnset{(\rA\to\rB)}$ 
and for every $k\in\rZ$ a post-processing test 
$\test{P}^{(k)}\coloneqq \{\tP_l^{(k)}\}_{l\in\rL}\subseteq\Trnset{(\rB\to\rC)}$ such that
\begin{subequations}
 \begin{align}\label{eq:obs1}
 \Qcircuit @C=0.7em @R=1em 
  {&
  \qw\poloFantasmaCn{\rA} &\measureD{a_i}
                            }
                            \; &= \;
  \sum_{k\in\rV_i} 
  \begin{aligned}
    \Qcircuit @C=1em @R=.7em @! R { 
    & \qw\poloFantasmaCn{\rA}
      &\gate{{\tD}_{k}} &\qw\poloFantasmaCn{\rB}&\measureD{e}}
  \end{aligned},\\\label{eq:obs2}
 \Qcircuit @C=0.7em @R=1em 
  {&
  \qw\poloFantasmaCn{\rA} &\measureD{b_j}
                            }
                            \; &= \;   
  \sum_{(k,l)\in\rW_j} 
  \begin{aligned}
    \Qcircuit @C=1em @R=.7em @! R { 
      & \qw\poloFantasmaCn{\rA}
      &\gate{{\tD}_{k}} &\qw\poloFantasmaCn{\rB}
      &\gate{{\tP_l^{(k)}}}  &\qw\poloFantasmaCn{\rC}
      &\measureD{e}}
  \end{aligned}\,
\end{align}
\end{subequations}
for 
suitable disjoint partitions $\{\rV_i\}_{i\in\rX}$ of $\rZ$ and $\{\rW_{j}\}_{j\in\rY}$ of $\rZ\times\rL$. 
Now consider the observation test $\test{F}\coloneqq\{f_{k,l}\}_{(k,l)\in\rX\times\rY}$ defined as follows
\begin{align}\nonumber
 \Qcircuit @C=0.7em @R=1em 
  {&
  \qw\poloFantasmaCn{\rA} &\measureD{f_{k,l}}
                            }
                            \; &\coloneqq \;
\begin{aligned}
    \Qcircuit @C=1em @R=.7em @! R { 
      & \qw\poloFantasmaCn{\rA}
      &\gate{{\tD}_{k}} &\qw\poloFantasmaCn{\rB}
      &\gate{{\tP_l^{(k)}}}  &\qw\poloFantasmaCn{\rC}
      &\measureD{e}}
  \end{aligned}.
\end{align}
On the one hand by Eq.~\eqref{eq:obs2} it immediatly follows that $\sum_{(k,l)\in\rW_j} \B{f_{k,l}}=\B{b_j}$. On the other hand one has
\begin{align}\nonumber
\sum_{(k,l)\in\rV_i\times \rL} \Qcircuit @C=0.7em @R=1em 
  {&
  \qw\poloFantasmaCn{\rA} &\measureD{f_{k,l}}}\ &= \;
	\sum_{k\in\rV_i}
	\begin{aligned}
		\Qcircuit @C=1em @R=.7em @! R { 
      		& \qw\poloFantasmaCn{\rA}
      		&\gate{{\tD}_{k}} &\qw\poloFantasmaCn{\rB}
      		&\measureD{e}}
	\end{aligned}\\
&  = \;
   \Qcircuit @C=0.7em @R=1em 
  {&
  \qw\poloFantasmaCn{\rA} &\measureD{a_i}
                            }
\end{align}
where the first equality follows from Eq.~\eqref{eq:obs2} reminding that $\sum_l\tP^{(k)}_l$ are channels and then 
preserve the deterministic effect. Thus we have two different coarse-grainings of the same test $\test{F}$, given by  
$\sum_{(k,l)\in\rV_i\times\rL}  \B{f_{k,l}}=\B{a_i}$ and $\sum_{(k,l)\in\rW_j} \B{f_{k,l}}=\B{b_j}$, corresponding to 
the two observation tests $\test{A}$ and $\test{B}$, respectively. We can now simply define the observation test 
$\test{C}\coloneqq\{c_{i,j}\}_{(i,j)\in\rX\otimes\rY}\in\Cntset{(\rA)}$ with 
$\B{c_{i,j}}=\sum_{(k,l)\in(\rV_i\times\rL)\cap\rW_j}\B{f_{k,l}}$ such that $\ref{obs:4}$ is satisfied. This proves that 
$\test{A}\rightarrow\test{B}$ implies $\ref{obs:4}$ and exactly the same argument works assuming 
$\test{B}\rightarrow\test{A}$ in place of of $\test{A}\rightarrow\test{B}$. 
We conclude the proof showing that \ref{obs:4}$\Rightarrow$\ref{obs:1}. In the case of interest the condition in 
Eq.~\eqref{eq:strong-a-bis} simplifies, being systems $\rB$ and $\rC$ trivial. Thus, the condition boils down to 
having an observation test $\{c_{i,j}\}_{(i,j)\in\rX\times\rY}$ whose marginal coarse-grainings in 
Eq.~\eqref{eq:strong-a-bis} are equal to $a_i$ and $b_j$, respectively. However, this is precisely the hypothesis in 
item~\ref{obs:4}.
\end{proof}

As a final remark in this subsection, we notice that two observation tests are compatible if and only if they are the measurements associated to two compatible instruments, both in the weak and in the strong sense: 

\begin{proposition} 
The four conditions in Proposition~\ref{prop:compatible-observation-tests} are also equivalent to the following two
\begin{enumerate}
\setcounter{enumi}{4}
\item \label{it:strongtest}There exist two tests $\tilde{\test{A}}\coloneqq\{\tA_i\}_{i\in\rX}\subseteq\Trnset{(\rA,\rB)}$ and $\tilde{\test{B}}\coloneqq\{\tB_j\}_{j\in\rY}\subseteq\Trnset{(\rA,\rC)}$, with $\B{a_i}\coloneqq\B{e}\tA_i$ and $\B{b_j}\coloneqq\B{e}\tB_j$ such that $\tilde{\test{A}}\Join\tilde{\test{B}}$;
\item \label{it:weaktest}There exist two tests $\tilde{\test{A}}\coloneqq\{\tA_i\}_{i\in\rX}\subseteq\Trnset{(\rA,\rB)}$ and $\tilde{\test{B}}\coloneqq\{\tB_j\}_{j\in\rY}\subseteq\Trnset{(\rA,\rC)}$, with $\B{a_i}\coloneqq\B{e}\tA_i$ and $\B{b_j}\coloneqq\B{e}\tB_j$ such that $\tilde{\test{A}}\leftrightarrow\tilde{\test{B}}$.
\end{enumerate}
\end{proposition}
\begin{proof} The scheme of the proof is the following: $\ref{it:strongtest}\Rightarrow\ref{it:weaktest}\Rightarrow\ref{obs:1}\Leftrightarrow\ref{obs:4}\Rightarrow\ref{it:strongtest}$.
By Proposition~\ref{prop:stron-weak}, clearly $\ref{it:strongtest}\Rightarrow\ref{it:weaktest}$. Now, 
one can straightforwardly see that~$\ref{it:weaktest}\Rightarrow\ref{obs:1}$, taking the definition of weak compatibility for $\tilde{\test A}$ and $\tilde{\test B}$, and discarding both outputs. Finally, one can prove that~$\ref{obs:4}\Rightarrow\ref{it:strongtest}$ as follows. Define the tests $\tilde{\test{A}}\coloneqq\{\tA_i\}_{i\in\rX}\subseteq\Trnset{(\rA,\rB)}$, 
$\tilde{\test{B}}\coloneqq\{\tB_j\}_{j\in\rY}\subseteq\Trnset{(\rA,\rC)}$, and $\tilde{\test{C}}\coloneqq\{\tC_{i,j}\}_{(i,j)\in\rX\times\rY}\subseteq\Trnset{(\rA,\rB\rC)}$ as follows 
\begin{align*}
\Qcircuit @C=0.7em @R=1em 
    {&
       \qw\poloFantasmaCn{\rA}& \gate{\tA_i}& \qw\poloFantasmaCn{\rB}&\qw
                                                                       }&\coloneqq
\Qcircuit @C=0.7em @R=1em 
  {&
  \qw\poloFantasmaCn{\rA} &\measureD{a_i}
                            }                                                                      
  \Qcircuit @C=.5em @R=.5em { 
  &\prepareC{\omega} &
      \poloFantasmaCn{\rB} \qw&\qw 
    }\;,
                                            \\                            \\
\Qcircuit @C=0.7em @R=1em 
    {&
       \qw\poloFantasmaCn{\rA}& \gate{\tB_j}& \qw\poloFantasmaCn{\rC}&\qw
                                                                       }&\coloneqq
\Qcircuit @C=0.7em @R=1em 
  {&
  \qw\poloFantasmaCn{\rA} &\measureD{b_j}
                            }                                                                      
  \Qcircuit @C=.5em @R=.5em { 
  &\prepareC{\nu} &
      \poloFantasmaCn{\rC} \qw&\qw 
    }\;,
                                                                       \\ \\
  \begin{aligned}
    \Qcircuit @C=1em @R=.7em @! R { 
    & \qw\poloFantasmaCn{\rA}
      &\multigate{1}{{\tC}_{i,j}} &\qw\poloFantasmaCn{\rB}&\qw 
      \\
      & & \pureghost{{\tC}_{i,j}} &\qw\poloFantasmaCn{\rC}
      &\qw}
  \end{aligned}&\coloneqq
\Qcircuit @C=0.7em @R=1em 
  {&
  \qw\poloFantasmaCn{\rA} &\measureD{c_{i,j}}
                            }\;                                                                      
\begin{aligned} 
\Qcircuit @C=1em @R=.7em @! R {
\prepareC{\omega}&\qw \poloFantasmaCn{\rB}&\qw\\  
\prepareC{\nu}&\qw \poloFantasmaCn{\rC}&\qw  }
\end{aligned}
\;,
\end{align*}
where $\{\omega\}\subseteq\Stset(\rB)$, and $\{\nu\}\subseteq\Stset(\rC)$ are two arbitrary deterministic preparation tests. It immediately follows that $\B{a_i}=\B{e}\tA_i$ and $\B{b_j}=\B{e}\tB_j$. Using $\tilde{\test{C}}$ for the joint simulation of $\tilde{\test{A}}$ and $\tilde{\test{B}}$, one then has  $\tilde{\test{A}}\Join\tilde{\test{B}}$.
\end{proof}

\subsubsection{The case of channels}

Also the compatibility between channels, i.e.~deterministic tests made of a single event, is a special case of
the one given for general tests. However, in this case strong and weak compatibility do not lead to the same notion as it was for observables. 
For the convenience of the reader we  rewrite Definitions~\ref{def:strong} and~\ref{def:weak-fin} in the case of channels.

According to Definition~\ref{def:strong} two channels $\tA\in\Trnset{(\rA,\rB)}$ and $\tB\in\Trnset{(\rA,\rC)}$ are strongly compatible if there exists a channel $\tC\in\Trnset{(\rA,\rB\rC)}$ such that
\begin{subequations}
\begin{align}\label{eq:strong-comp-channels-a}
& \Qcircuit @C=0.7em @R=1em 
    {&
       \qw\poloFantasmaCn{\rA}& \gate{\tA}& \qw\poloFantasmaCn{\rB}&\qw
                                                                       }
                                                                       \; =
   \begin{aligned}
    \Qcircuit @C=1em @R=.7em @! R { 
    & \qw\poloFantasmaCn{\rA}
      &\multigate{1}{{\tC}} &\qw\poloFantasmaCn{\rB}&\qw
      \\
      & & \pureghost{{\tC}} &\qw\poloFantasmaCn{\rC}
      &\measureD{e} }
  \end{aligned},\\\label{eq:strong-comp-channels-b}
& \Qcircuit @C=0.7em @R=1em 
    {&
       \qw\poloFantasmaCn{\rA}& \gate{\tB}& \qw\poloFantasmaCn{\rC}&\qw
                                                                       }
                                                                       \; =
  \begin{aligned}
    \Qcircuit @C=1em @R=.7em @! R { 
    & \qw\poloFantasmaCn{\rA}
      &\multigate{1}{{\tC}} &\qw\poloFantasmaCn{\rB}&\measureD{e} 
      \\
      & & \pureghost{{\tC}} &\qw\poloFantasmaCn{\rC}
      &\qw}
  \end{aligned}.
\end{align}
\end{subequations}
This is the criterion studied for quantum channels in Refs.~\cite{Teiko-incompatibility-2016,Teiko-incompatibility-2017}. On the other hand, based on Definition~\ref{def:weak}, $\tA$ and $\tB$ are weakly compatible if $\tA\rightarrow\tB$ and $\tB\rightarrow\tA$, namely
\begin{subequations}
\begin{align}\label{eq:weak-comp-channels-a}
& \Qcircuit @C=0.7em @R=1em 
    {&
       \qw\poloFantasmaCn{\rA}& \gate{\tA}& \qw\poloFantasmaCn{\rB}&\qw
                                                                       }
                                                                       \; =
  \sum_{k} 
  \begin{aligned}
    \Qcircuit @C=1em @R=.7em @! R { 
    & \qw\poloFantasmaCn{\rA}
      &\multigate{1}{{\tC}_{k}} &\qw\poloFantasmaCn{\rB}&\qw
      \\
      & & \pureghost{{\tC}_{k}} &\qw\poloFantasmaCn{\rB'}
      &\measureD{e} }
  \end{aligned},
  \\\label{eq:weak-comp-channels-b}
&\Qcircuit @C=0.7em @R=1em 
    {&
       \qw\poloFantasmaCn{\rA}& \gate{\tB}& \qw\poloFantasmaCn{\rC}&\qw
                                                                       }
                                                                       \; =
  \sum_{k} 
  \begin{aligned}
    \Qcircuit @C=1em @R=.7em @! R { 
      & \qw\poloFantasmaCn{\rA}
      &\multigate{1}{{\tC}_{k}} &\qw\poloFantasmaCn{\rB}
      &\multigate{1}{{\tP^{(k)}}}  &\qw\poloFantasmaCn{\rC}
      &\qw
      \\
      & 
      & \pureghost{{\tC}_{k}} &\qw\poloFantasmaCn{\rB'}
      &  \ghost{{\tP^{(k)}}} &\qw\poloFantasmaCn{\rC'}
      &\measureD{e}}
  \end{aligned}\,
\end{align}
\end{subequations}
for some test $\test{C}\coloneqq\{\tC_k\}_{k\in\rZ}\subseteq\Trnset{(\rA,\rB\rB')}$ and deterministic transformations $\tP^{(k)}\in\Trnset{(\rB\rB',\rC\rC')}$ (and analogous condition for $\tB\rightarrow\tA$).
The non-equivalence between the two conditions above is a consequence of the following lemmas, the first already proved in Ref.~\cite{Teiko-incompatibility-2017}.

\begin{lemma}[Broadcasting via strongly-compatible reversible channels]\label{lem:broadcasting}
If there exist two left-reversible channels $\tA\in\Trnset{(\rA,\rB)}$ and $\tB\in\Trnset{(\rA,\rC)}$ that are strongly compatible then there exists a broadcasting map for system $\rA$.
\end{lemma}

\begin{proof}
A broadcasting map for system $\rA$ is a transformation $\tC\in\Trnset(\rA,\rA\rA)$ such that, for every state $\rho\in\Stset(\rA\rD)$, the state $\tC\K{\rho}\in\Stset(\rA\rA\rD)$ has both marginals on system $\rA$ equal to $\rho$. Suppose now that $\tA\in\Trnset{(\rA,\rB)}$ and $\tB\in\Trnset{(\rA,\rC)}$ are left-reversible and strongly-compatible, namely there exists a channel $\tC\in\Trnset(\rA,\rB\rC)$ with marginals equal to $\tA$ and $\tB$ (see Eqs.~\eqref{eq:strong-comp-channels-a} and \eqref{eq:strong-comp-channels-b}). The  channel
\begin{align}\nonumber
  \begin{aligned}
    \Qcircuit @C=1em @R=.7em @! R { 
    & \qw\poloFantasmaCn{\rA}
      &\multigate{1}{{\tC}} &\qw\poloFantasmaCn{\rB}&\gate{\tV}&\qw\poloFantasmaCn{\rA}&\qw
      \\
      & & \pureghost{{\tC}} &\qw\poloFantasmaCn{\rC}&\gate{\tW}&\qw\poloFantasmaCn{\rA}&\qw}
  \end{aligned},
\end{align}
with $\tV$ and $\tW$ such that $\tV\tA=\tI_\rA$ and $\tW\tB=\tI_\rA$ respectively, has both marginals equal to the identity channel, thus providing a broadcasting channel.
\end{proof}

\begin{lemma}[Weak-compatibility of channels with reversible dilation]\label{lem:compatible-channels-dilation}
Two channels with reversible dilation are weakly-compatible.
\end{lemma}
\begin{proof}
Consider two channels $\tA\in\Trnset{(\rA,\rB)}$ and $\tB\in\Trnset{(\rA,\rC)}$, and suppose that they have a left-
reversible dilation, namely there exist reversible transformations $\tV\in\Trnset{(\rA,\rB\tilde\rB)}$ and $\tW\in\Trnset{(\rA,\rC\tilde\rC)}$, such that
  \begin{align}
  &   \nonumber\Qcircuit @C=0.7em @R=1em 
    {&
       \qw\poloFantasmaCn{\rA}& \gate{\tA}& \qw\poloFantasmaCn{\rB}&\qw
                                                                       }
                                                                       \; =\; 
\begin{aligned}
    \Qcircuit @C=1em @R=.7em @! R { & \qw\poloFantasmaCn{\rA}
      &\multigate{1}{\tV} &\qw\poloFantasmaCn{\rB}
      &\qw\\
      &&\pureghost{\tV} &\qw\poloFantasmaCn{\tilde\rB}
      &\measureD{e} }
  \end{aligned}\,,
\\
 &   \nonumber\Qcircuit @C=0.7em @R=1em 
    {&
       \qw\poloFantasmaCn{\rA}& \gate{\tB}& \qw\poloFantasmaCn{\rC}&\qw
                                                                       }
                                                                       \; =\; 
\begin{aligned}
    \Qcircuit @C=1em @R=.7em @! R { & \qw\poloFantasmaCn{\rA}
      &\multigate{1}{\tW} &\qw\poloFantasmaCn{\rC}
      &\qw\\
      &&\pureghost{\tW} &\qw\poloFantasmaCn{\tilde\rC}
      &\measureD{e} }
  \end{aligned}\,.
  \end{align}
Now it is immediate to see that $\tA\rightarrow\tB$: indeed one can take $\tV$ in place of test $\test{C}$ in Eq.~\eqref{eq:weak-comp-channels-a} and $\tW\tilde\tV$ as post-processing in Eq.~\ref{eq:weak-comp-channels-b}, where $\tilde\tV$ is a left-inverse of $\tV$. Analogously, one can check that $\tB\rightarrow\tA$ using first $\tW$ to realize $\tB$, by discarding system $\tilde\rC$ and $\tV\tilde\tW$, with $\tilde\tW$ a left-inverse of $\tW$, as post-processing to get $\tA$ after discarding system $\tilde\rB$.
\end{proof}

\begin{corollary} Strong compatibility is strictly stronger than weak compatibility. 
\end{corollary}
\begin{proof} This follows from the case of quantum channels. Since quantum theory satisfies no-broadcasting \cite{PhysRevLett.76.2818}, no pair of reversible channels can be strongly compatible, according to Lemma~\ref{lem:broadcasting}. On the other hand, quantum theory allows for reversible dilation of channels and thus, due to Lemma~\ref{lem:compatible-channels-dilation}, every pair of channels is weakly-compatible in quantum theory.
\end{proof}

\section{Theories with no incompatible tests}\label{sec:IV}
In this section we characterize theories that do not exhibit any pair of incompatible tests in the weak (and then also in the strong) sense. We call this feature full compatibility and show necessary and sufficient conditions for it. Among them is the equivalence between the class of operational probabilistic theories with full compatibility and the class of strongly causal theories with full information without disturbance, namely theories where all information can be extracted without introducing disturbance. This last property has been presented in Ref.~\cite{DAriano2020information}, where the authors prove criteria for studying the information-disturbance trade-off in an arbitrary operational probabilistic theory. A brief account of the results in Ref.~\cite{DAriano2020information} is given in the next section.

\subsection{Information-disturbance relation}

Probabilistic theories can be divided in three classes according their information-disturbance trade-off: 
\begin{enumerate}[label=(\roman*)]
\item\emph{No-information without disturbance}: this class encompasses theories where any test that provides 
information must disturb. This is the case e.g.~of quantum 
theory~\cite{RevModPhys.86.1261,Busch2009,dariano_chiribella_perinotti_2017}, or non-local boxes, as proved in 
Ref.~\cite{DAriano2020information}.\label{case1}
\item\emph{Information without disturbance}: this is the class corresponding to the most typical scenario, where some 
information can be extracted without disturbance and some other information cannot.\label{case2}
\item\emph{Full-information without disturbance}: this is the class of theories where all information can be extracted 
without disturbance. The prototype of theories in this class is classical information theory.\label{case3}
\end{enumerate}

This classification is based on the following definition of disturbance that does not rely on any specific feature of the  theory.

\begin{definition}[Non-disturbing test]\label{def:non-disturbing test}
   A test $\test{A}\coloneqq\{\tA_i\}_{i\in\rX}\subseteq\Trnset{(\rA,\rB)}$ is {\em non-disturbing} 
   if  $\sum_i\tA_i=\tR$, where $\tR\in\Trnset{(\rA,\rB)}$ is a reversible transformation. In a (weakly) causal theory 
   we require $\tR$ to be just left-reversible.
\end{definition}  
  Indeed, if a test is coarse-grained to a reversible transformation, then
  its effect can be corrected by inverting it. Restricting to the scenario of (weakly) causal theories, it is 
  reasonable
  to relax the above definition including the possibility that the channel $\tR$ is just left-reversible, meaning 
  that the correction can be applied \emph{after} the test. The above condition in the (weakly) causal case 
  states that a test is non-disturbing if its effect can be corrected not only on local states $\rho\in\Stset{(\rA)}$, 
  but also when it is locally applied to every state $\Psi\in\Stset{(\rA\rC)}$ for every ancillary system $\rC$.  This 
  criterion captures the meaning of disturbance also in theories without local tomography, where a test that preserves 
  all local states $\rho\in\Stset{(\rA)}$, still could modify a state $\Psi\in\Stset{(\rA\rC)}$ for some system $\rC$. 
  An example of this phenomenon occurs in fermionic quantum theory \cite{Bravyi2002210,d2014feynman,D_Ariano_2014} as 
  shown in Ref.~\cite{DAriano2020information}.

A central role in the above classification is played by the identity transformation as established by the following proposition, that summarises various results proved in Ref.~\cite{DAriano2020information}
\begin{proposition}[Atomic decomposition of the identity]\label{prop:identity-atomic-decomposition} 
Given an operational probabilistic theory one of the following holds:
\begin{enumerate}
\item The identity transformation is atomic for every system. \label{it:atid}
This is a necessary and sufficient condition for  no-information without disturbance.

\item \label{it:nonatid}The identity is not atomic for some system. The non redundant atomic refinement $\{\tA_i\}_{i\in\rX}\subseteq\Trnset(\rA)$ of the identity
  $\tI_\rA=\sum_i\tA_i$ is unique with $\tA_i\tA_j=\tA_i\delta_{ij}$. Thus, the sets of states and effects of system $\rA$ decompose as
  \begin{align}\label{eq:sumst}
  \begin{aligned}
  \Stset(\rA)&=\bigoplus_{i\in\rX}\Stset_{i}(\rA),\\
  \Cntset(\rA)&=\bigoplus_{i\in\rX}\Cntset_{i}(\rA),
\end{aligned}
  \end{align}
where $\Stset_{i}(\rA)\coloneqq\tA_i\Stset(\rA)$ and $\Cntset_{i}(\rA)\coloneqq\Cntset(\rA)\tA_i$.
  Moreover, for any pair of systems $\rA$, $\rB$, with non redundant atomic decompositions $\{\tA_i\}_{i\in\rX}$, $\{\tB_j\}_{j\in\rY}$ of the identities $\tI_\rA$ and $\tI_\rB$, one
has the following decomposition of the set of states and of the set of effects of $\rA\rB$
\begin{equation}\label{bigoplus}
\begin{aligned}
  \Stset(\rA\rB)&=\bigoplus_{(i,j)\in\rX\times\rY}\bigoplus_{k\in K(i,j)}\Stset_{ijk}(\rA\rB),\\
  \Cntset(\rA\rB)&=\bigoplus_{(i,j)\in\rX\times\rY}\bigoplus_{k\in K(i,j)}\Cntset_{ijk}(\rA\rB),
\end{aligned}
\end{equation}  
where
\begin{equation}\label{bigolus}
\begin{aligned}
  \tA_i\otimes\tA_j\Stset(\rA\rB)&=\bigoplus_{k\in K(i,j)}\Stset_{ijk}(\rA\rB),\\
  \Cntset(\rA\rB)\tA_i\otimes\tA_j&=\bigoplus_{k\in K(i,j)}\Cntset_{ijk}(\rA\rB).
\end{aligned}
\end{equation}  
Clearly, 
$\tA_i\K{\Psi_{i'}}=\delta_{ii'}\K{\Psi_{i}}$, and $\B{c_{j'}} \tA_j=\delta_{jj'}\B{c_{j}}$
for all $\Psi_{i'}\in\Stset_{i'}(\rA)$ and
$c_{j'}\in\Cntset_{j'}(\rA)$.
\end{enumerate}
\end{proposition}
In case~\ref{it:nonatid} there is always some information that can be extracted without disturbance---from Proposition~\ref{prop:identity-atomic-decomposition}, this is clearly the 
information corresponding to the outcomes of tests refining the identity transformation. As a special case one has the 
class of theories with full-information without disturbance (case~\ref{case3}), where any test can be realized in such 
a way that via a post-processing its effects can be ``erased''.
\begin{definition}[Full-information without disturbance]\label{def:fiwd}
A theory satisfies full-information without disturbance if for every
$\test{A}\coloneqq\{\tA_i\}_{i\in\rX}\subseteq\Trnset{(\rA,\rB)}$ one has $\test{A}\rightarrow\tI_\rA$.
\end{definition}

\subsection{Necessary and sufficient conditions for full-compatibility} 
We are now interested in characterizing the class of theories that do not include incompatible tests. We then introduce the class of fully-compatible theories as the class of theories satisfying the following feature:
\begin{definition}[Full-compatibility]
  A theory satisfies full-compatibility if for every pair of tests
  $\test{A}\subseteq\Trnset{(\rA,\rB)}$ and $\test{B}\subseteq\Trnset{(\rA,\rC)}$,
  one has $\test{A}\leftrightarrow\test{B}$.
\end{definition}

We can now prove equivalent conditions for full-compatibility. We show that the class of theories satisfying full 
compatibility coincides with the class of theories with full-information without disturbance. Moreover, we show that 
full compatibility is equivalent to the condition that any test can be simulated using a non-disturbing test,  possibly involving an ancillary system.
Before proving the main result we present a preliminary lemma.

\begin{lemma}\label{lem:discrim}
In any theory violating no-information without disturbance there exist systems with an arbitrarily large number of perfectly discriminable states.
 \end{lemma}
\begin{proof}
According to Proposition~\ref{prop:identity-atomic-decomposition} in any theory where no-information without 
disturbance does not hold there exists a system $\rA$ such that the identity $\tI_\rA$ is not atomic. 
Let $\{\tA_i\}_{i\in\rX}$, with $|\rX|=m\geq2$, be the unique non redundant atomic decomposition of $\tI_\rA$. 
Therefore, 
again by Proposition~\ref{prop:identity-atomic-decomposition}, $\{\tA_{i_1i_2\ldots i_n}\}_{(i_1i_2\ldots i_n)\in\rX^n}$, with $\tA_{i_1i_2\ldots i_n}\coloneqq\tA_{i_1}\otimes\tA_{i_2}\otimes\ldots\otimes\tA_{i_n}$,
is a decomposition of the identity of system $\rN\coloneqq\rA^{\otimes n}$, and the set of states $\Stset{(\rN)}$ splits in $mn$ blocks, i.e.~$\Stset(\rN)=\bigoplus_{(i_1i_2\ldots i_n)\in\rX^n}\Stset_{i_1i_2\ldots i_n}{(\rN)}$---where $\Stset_{i_1i_2\ldots i_n}{(\rN)}$ in turn splits as $\Stset_{i_1i_2\ldots i_n}{(\rN)}=\bigoplus_{k\in K(i_1i_2\ldots i_n)}\Stset_{i_1i_2\ldots i_Nk}(\rN)$. Since each block spans at least a one-dimensional vector space, there exists a set of non-null states 
$\{\tilde{\rho}_{i_1i_2\ldots i_n}\}_{(i_1i_2\ldots i_n)\in\rX^n}$ such that $\tA_{i'_1i'_2\ldots i'_n}\K{\tilde{\rho}_{i_1i_2\ldots i_n}}=\delta_{i_1i'_1}\delta_{i_2i'_2}\ldots\delta_{i_ni'_n}\K{\tilde{\rho}_{i_1i_2\ldots i_n}}$. Moreover, since the theory is strongly causal and each state is proportional to a deterministic one, one can always consider the set of deterministic states $\{\rho_{i_1i_2\ldots i_n}\}_{(i_1i_2\ldots i_n)\in\rX^n}$, where $\K{\rho_{i_1i_2\ldots i_n}}=\K{\tilde{\rho}_{i_1i_2\ldots i_n}}/(e \K{\tilde{\rho}_{i_1i_2\ldots i_n}}$, where $e$ 
denotes the deterministic effect of system $\rN$. The last set still satisfies the relation $\tA_{i'_1i'_2\ldots i'_n}\K{\rho_{i_1i_2\ldots i_n}}=\delta_{i_1i'_1}\delta_{i_2i'_2}\ldots\delta_{i_ni'_n}\K{\rho_{i_1i_2\ldots i_n}}$. Finally, we observe that the observation test 
$\{a_{i_1i_2\ldots i_n}\}_{(i_1i_2\ldots i_n)\in\rX^n}$, with $\B{a_{i_1i_2\ldots i_n}}\coloneqq\B{e}\tA_{i_1i_2\ldots i_n}$, is such that 
$(a_{i_1i_2\ldots i_n}|\rho_{i'_1i'_2\ldots i'_n})=\delta_{i_1i'_1}\delta_{i_2i'_2}\ldots\delta_{i_ni'_n}$.
\end{proof}

We are now ready to prove the main result in this section.
\begin{theorem}
The following conditions are equivalent
  \begin{enumerate}
\item\label{i:fc} The theory satisfies full compatibility
\item\label{i:fiwd} The theory satisfies full information without disturbance
\item\label{i:universality} For every test 
$\test{A}\coloneqq\{\tA_i\}_{i\in\rX}\subseteq\Trnset{(\rA,\rB)}$ one has
  \begin{equation*}
\begin{aligned}
  & \Qcircuit @C=0.7em @R=1em {& \qw\poloFantasmaCn{\rA}&
    \gate{\tA_i}& \qw\poloFantasmaCn{\rB}&\qw }
    \ =\;
    \sum_{k\in\rZ_i} 
\begin{aligned}
  \Qcircuit @C=1em @R=.7em @! R {&
    \qw\poloFantasmaCn{\rA}&
    \multigate{1}{\tV_k}& \qw\poloFantasmaCn{\rB}&\qw \\
    &&\pureghost{\tV_k}&\qw\poloFantasmaCn{\tilde\rB}&\measureD{e}}
\end{aligned}\,,
\end{aligned}
\end{equation*}  
with $\test{V}\coloneqq\{\tV_k\}_{k\in\rZ}\subseteq\Trnset{(\rA,\rB)}$ such that $\sum_k\tV_k$ is a left-reversible channel.

\end{enumerate}
\end{theorem}
\begin{proof}
We prove the chain of implications \ref{i:fc} $\Rightarrow$ \ref{i:fiwd} $\Rightarrow$ \ref{i:universality} $\Rightarrow$ \ref{i:fc}. 
That $\ref{i:fc}\Rightarrow \ref{i:fiwd}$ trivially follows from Definition~\ref{def:weak}, considering 
$\test B=\{\tI_\rA\}$. As to $\ref{i:fiwd}\Rightarrow\ref{i:universality}$, one can invoke Lemma~\ref{lem:discrim}
along with strong causality, and considering the schemes~\eqref{eq:weak-a} and~\eqref{eq:weak-b}, with 
$\{\tB_j\}_{j\in\rY}\equiv\{\tI_\rA\}$, we can construct the test
\begin{align}
  \begin{aligned}
    \Qcircuit @C=1em @R=.7em @! R { & \qw\poloFantasmaCn{\rA}
      &\multigate{2}{\tW} & \qw\poloFantasmaCn{\rB}
      &\qw\\
      & &\pureghost{\tW} &\qw\poloFantasmaCn{\rB^\prime}
      &\qw \\
      &&\pureghost{\tW}&\qw\poloFantasmaCn{\rB''}
      &\qw}      
  \end{aligned}\ \coloneqq\ 
  \sum_{k\in\rZ} \; 
  \begin{aligned}
    \Qcircuit @C=1em @R=.7em @! R { & \qw\poloFantasmaCn{\rA}
      &\multigate{1}{\tC_{k}} & \qw\poloFantasmaCn{\rB}
      &\qw\\
      & &\pureghost{\tC_k}&\qw\poloFantasmaCn{\rB'}
      &\qw \\
      &&\prepareC{\psi_{k}}&\qw\poloFantasmaCn{\rB''}
      &\qw}
  \end{aligned}\ ,
\end{align}
that is reversible by
\begin{align}
  \begin{aligned}
    \Qcircuit @C=1em @R=.7em @! R { & \qw\poloFantasmaCn{\rB}
      &\multigate{2}{\tZ} & \qw\poloFantasmaCn{\rA}
      &\qw\\
      &\qw\poloFantasmaCn{\rB^\prime} &\ghost{\tZ} &
      & \\
      &\qw\poloFantasmaCn{\rB''}&\ghost{\tZ}&
      &}      
  \end{aligned}\ \coloneqq\ 
  \sum_{k}  \; 
  \begin{aligned}
    \Qcircuit @C=1em @R=.7em @! R { & \qw\poloFantasmaCn{\rB}
      &\multigate{1}{\tP^{(k)}} & \qw\poloFantasmaCn{\rA}&\qw\\
      &\qw\poloFantasmaCn{\rB'}&\ghost{{\tP}^{(k)}} 
      &\qw\poloFantasmaCn{\rC'}&\measureD{e} \\
      &\qw\poloFantasmaCn{\rB''}&\measureD{a_{k}}&&&
      &}
  \end{aligned}\ ,
\end{align}
where $\{\psi_k\}_{k\in\rZ}$ is a set of perfectly discriminable deterministic states with 
$\K{\psi_k}\in\Stset_k(\rB'')$,
$\B{a_k}\coloneqq\B{e}\tilde\tA_k$, $e$ denoting the deterministic effect for system ${\rB''}$, and 
$\{\tilde\tA_k\}_{k\in\rZ}$ is the atomic non-redundant decomposition of $\tI_{\rB''}$. Finally, we define 
$\tilde\rB\coloneqq\rB'\rB''$, and
\begin{align*}
\begin{aligned}
\Qcircuit @C=0.7em @R=1em 
    {&
       \qw\poloFantasmaCn{\rA}& \multigate{1}{\tV_k}& \qw\poloFantasmaCn{\rB}&\qw\\
       && \pureghost{\tV_k}& \qw\poloFantasmaCn{\tilde\rB}&\qw}
\end{aligned}  \; \coloneqq
  \begin{aligned}
    \Qcircuit @C=1em @R=.7em @! R { & \qw\poloFantasmaCn{\rA}
      &\multigate{2}{\tW} &\qw& \qw\poloFantasmaCn{\rB}
      &\qw&\qw\\
      & &\pureghost{\tW} &\qw&\qw\poloFantasmaCn{\rB^\prime}
      &\qw&\qw \\
      &&\pureghost{\tW}&\qw\poloFantasmaCn{\rB''}
      &\gate{\tilde\tA_{k}}&\qw\poloFantasmaCn{\rB''}&\qw}      
  \end{aligned}\ .
\end{align*}
We now show that  \ref{i:universality} $\Rightarrow$ \ref{i:fc}. Indeed, given two tests $\test{A}\coloneqq \{\tA_i\}_{i\in\rX}\subseteq\Trnset{(\rA,\rB)}$ and
 $\test{B}\coloneqq\{\tB_j\}_{j\in\rY}\subseteq\Trnset{(\rA,\rC)}$, according to \ref{i:universality} one has
\begin{align}
  & \Qcircuit @C=0.7em @R=1em {& \qw\poloFantasmaCn{\rA}&
    \gate{\tA_i}& \qw\poloFantasmaCn{\rB}&\qw }
    \ =\;
    \sum_{k\in\rZ_i} 
 \begin{aligned}
  \Qcircuit @C=1em @R=.7em @! R {&
    \qw\poloFantasmaCn{\rA}&
    \multigate{1}{\tV_k}& \qw\poloFantasmaCn{\rB}&\qw \\
    &&\pureghost{\tV_k}&\qw\poloFantasmaCn{\tilde\rB}&\measureD{e}}
\end{aligned}\,,\\
  & \Qcircuit @C=0.7em @R=1em {& \qw\poloFantasmaCn{\rA}&
    \gate{\tB_j}& \qw\poloFantasmaCn{\rC}&\qw }
    \ =\;
    \sum_{l\in\rL_j} 
 \begin{aligned}
  \Qcircuit @C=1em @R=.7em @! R {&
    \qw\poloFantasmaCn{\rA}&
    \multigate{1}{\tW_l}& \qw\poloFantasmaCn{\rC}&\qw \\
    &&\pureghost{\tW_l}&\qw\poloFantasmaCn{\tilde\rC}&\measureD{e}}
\end{aligned}\,,
\end{align}  
with $\test{V}\coloneqq\{\tV_k\}_{k\in\rZ}\subseteq\Trnset{(\rA,\rB\tilde\rB)}$ and 
$\test{W}\coloneqq\{\tW_l\}_{l\in\rL}\subseteq\Trnset{(\rA,\rC\tilde\rC)}$
such that $\sum_k\tV_k=\tV$ and $\sum_l\tW_l=\tW$ are a left-reversible channels. Now to see that 
$\test{A}\rightarrow\test{B}$, in Eq.~\eqref{eq:weak-a} we consider the test $\test{C}$ equal to the test $\test{V}$, 
with $\rB'$ equal to $\tilde\rB$, and in Eq.~\eqref{eq:weak-b} we take the map
 \begin{align}
 \begin{aligned}
    \Qcircuit @C=1em @R=.7em @! R { 
      & \qw\poloFantasmaCn{\rB}
      &\multigate{1}{\tZ} & \qw\poloFantasmaCn{\rA}
      &\multigate{1}{\tW_l}&\qw\poloFantasmaCn{\rC}
      &\qw
      \\
      &\qw\poloFantasmaCn{\tilde\rB}
      &\ghost{\tZ}&
      &\pureghost{\tW_l}&\qw\poloFantasmaCn{\tilde\rC}
      &\qw}
\end{aligned},
 \end{align}
with $\tZ\in\Trnset{(\rB\tilde\rB,\rA)}$ the left-inverse of $\tV$, as post-processing map and sum over $(k,l)\in\rZ\times\rL_j$. Analogously one can show that also $\test{B}\rightarrow\test{A}$ holds.
\end{proof}

\begin{corollary}
A theory has some information that cannot be extracted without disturbance if and only if 
there exist incompatible tests.
\end{corollary}
\begin{proof}
This trivially follows from the equivalence between full-compatibility and full-information without disturbance.
\end{proof}


We conclude this Section proving that in a theory without incompatibility all the systems must be classical according to the following definition

\begin{definition}[Classical systmes] A system $\rA$, with linear dimension $d=\dim{(\Stset_\Reals(\rA))}=\dim{(\Cntset_\Reals(\rA))}$, is classical if 
the unique atomic non-redundant refinement of the identity $\{\tA_i\}_{i\rX}$ has cardinality $|\rX|=d$ and each block in the decomposition of Eq.~\eqref{eq:sumst} is one dimensional, namely $\dim(\tA_i\Stset_\Reals(\rA))=\dim(\tA_i\Cntset_\Reals(\rA))=1$ for every $i\in\rX$.
\end{definition}
We observe that according to the above definition a system is classical if and only if all its pure states are jointly 
perfectly discriminable. In this case the base of the conic hull of the set of states of each system is a simplex, 
which corresponds to a subset of the set of states for a convex theory. For a (weakly) causal theory, the set of 
states is a simplex itself, with $d+1$ vertices represented by the null state $0\in\Stset(\rA)$ and the $d$ 
deterministic pure states. An example of theory with such systems is obviously the customary classical information 
theory. However, even a theory where all systems are classical could differ from classical 
information theory, in particular in the rule for parallel composition of systems (see Ref.~\cite{PhysRevA.101.042118} 
for an example).

\begin{lemma}\label{lem:atomic-maps}
Let $\{\tA_i\}_{i\in\rX}\in\Trnset{(\rA)}$ and $\{\tB_j\}_{j\in\rY}\in\Trnset{(\rA)}$ be the unique non-redundant atomic decomposition of $\tI_\rA$ and $\tI_\rB$, respectively. If $\tC\in\Trnset{(\rA,\rB)}$ is atomic then one has $\tC=\tA_i\tC\tB_j$ for some $i\in\rX$ and $j\in\rY$.
\end{lemma}
\begin{proof}
Since $\sum_{i,j}\tA_i\tC\tB_j=\tC$ and $\tC$ is atomic one has $\tA_i\tC\tB_j=\lambda_{ij}\tC$, with $\lambda_{ij}\geq 0$, for every $i\in\rX$ and $j\in\rY$. On the other hand it is $\tA_k\tA_i\tC\tB_j\tB_l=\delta_{ik}\delta_{jl}\tA_i\tC\tB_j$ showing that $\lambda_{ij}=1$ for one choice of $i,j$ while it is null otherwise.
\end{proof}

\begin{proposition}
In a theory with full compatibility all system of the theory are classical.
\end{proposition}
\begin{proof}
Consider an atomic effect $a\in\Cntset{(\rA)}$. If $\{\tB_j\}_{j\in\rY}$ is the unique non-redundant decomposition of the identity of system $\rA$, then it must be 
$\B{a}\tB_{\bar j}=\B{a}$ for some ${\bar j}\in\rY$.  Let $\tA\in\Trnset{(\rA,\rB)}$ be an atomic transformation such that $\B{e}\tA=\mu\B{a}$, with $e$ the deterministic effect of system $\rA$ and $1\geq\mu>0$. Since $\tA$ is atomic, due to Lemma~\ref{lem:atomic-maps}, it must be $\B{e}\tA=\B{e}\tA\tB_{\bar j}$.
Consider now a test $\test{A}\coloneqq\{\tA_i\}_{i\in\rX}$ such that $\tA=\tA_{\bar i}$ for some $\bar i\in\rX$. Since by hypothesis all tests are weakly compatible,  it is $\test{A}\rightarrow\tI$, and using Definition~\ref{def:weak} one has 
\begin{subequations}\label{eq:one-dim-a}
 \begin{align}
&\Qcircuit @C=0.7em @R=1em 
    {&\qw\poloFantasmaCn{\rA}& \gate{\tB_{\bar j}}&
       \qw\poloFantasmaCn{\rA}& \gate{\tA_{\bar i}}& \qw\poloFantasmaCn{\rB}&\qw
                                                                       }
                                                                       \; 
=\;  \sum_{k\in\rZ_{\bar i}} 
  \begin{aligned}
    \Qcircuit @C=1em @R=.7em @! R { 
    &\qw\poloFantasmaCn{\rA}& \gate{\tB_{\bar j}}& \qw\poloFantasmaCn{\rA}
      &\multigate{1}{{\tC}_{k}} &\qw\poloFantasmaCn{\rB}&\qw
      \\
     && & & \pureghost{{\tC}_{k}} &\qw\poloFantasmaCn{\rB'}
      &\measureD{e} }
  \end{aligned},\\\label{eq:one-dim-b}
&\Qcircuit @C=0.7em @R=1em 
    {&
       \qw\poloFantasmaCn{\rA}& \gate{\tB_{\bar j}}&\qw\poloFantasmaCn{\rA}& \gate{\tI}& \qw\poloFantasmaCn{\rA}&\qw
                                                                       }
                                                                       \; \nonumber\\
                                                                      & \qquad= \sum_{k} 
  \begin{aligned}
    \Qcircuit @C=1em @R=.7em @! R {& 
      \qw\poloFantasmaCn{\rA}& \gate{\tB_{\bar j}}& \qw\poloFantasmaCn{\rA}
      &\multigate{1}{{\tC}_{k}} &\qw\poloFantasmaCn{\rB}
      &\multigate{1}{{\tP^{(k)}}}  &\qw\poloFantasmaCn{\rA}
      &\qw
      \\
      &
      &
      & 
      & \pureghost{{\tC}_{k}} &\qw\poloFantasmaCn{\rB'}
      &  \ghost{{\tP_j^{(k)}}} &\qw\poloFantasmaCn{\rA'}
      &\measureD{e}}
  \end{aligned}\,
\end{align}
\end{subequations}
for some test $\test{C}=\{\tC_k\}_{k\in\rZ}$ and channels $\tP^{(k)}$. By atomicity of $\tB_{\bar j}$ and Eq.~\eqref{eq:one-dim-b} it follows that for every $k$
\begin{equation}\label{eq:one-dim-c}
  \begin{aligned}
    \Qcircuit @C=1em @R=.7em @! R {& 
      \qw\poloFantasmaCn{\rA}& \gate{\tB_{\bar j}}& \qw\poloFantasmaCn{\rA}
      &\multigate{1}{{\tC}_{k}} &\qw\poloFantasmaCn{\rB}
      &\multigate{1}{{\tP^{(k)}}}  &\qw\poloFantasmaCn{\rA}
      &\qw
      \\
      &
      &
      & 
      & \pureghost{{\tC}_{k}} &\qw\poloFantasmaCn{\rB'}
      &  \ghost{{\tP_j^{(k)}}} &\qw\poloFantasmaCn{\rA'}
      &\measureD{e}}
  \end{aligned}=\lambda_k\tB_{\bar j},
\end{equation}
for some $\lambda_k\geq 0$. Since $\mu\B{a}=\B{e}\tA_{\bar i}\tB_{\bar j}$, Eq.~\eqref{eq:one-dim-a} implies that
\begin{equation}
\mu\ 
\begin{aligned}
\Qcircuit @C=0.7em @R=1em 
    {&\qw\poloFantasmaCn{\rA}& \measureD{a}}
\end{aligned}                                                                       \; 
=\;  \sum_{k\in\rZ_{\bar i}} 
  \begin{aligned}
    \Qcircuit @C=1em @R=.7em @! R { 
    &\qw\poloFantasmaCn{\rA}& \gate{\tB_{\bar j}}& \qw\poloFantasmaCn{\rA}
      &\multigate{1}{{\tC}_{k}} &\qw\poloFantasmaCn{\rB}&\measureD{e}
      \\
     && & & \pureghost{{\tC}_{k}} &\qw\poloFantasmaCn{\rB'}
      &\measureD{e} }
      \end{aligned},
\end{equation}
and using Eq.~\eqref{eq:one-dim-c}, remembering that channels preserve the deterministic effect, one finally gets 
$\B{a}\propto\B{e}\tB_{\bar j}$. 
Since this argument holds for all atomic effects in $\Cntset_{\bar j}{(\rA)}$, we conclude that 
$\Cntset_{\bar j}{(\rA)}$ has linear dimension equal to one.
 \end{proof}

\section{Discussion}
We have introduced two alternative definitions of compatible tests of a causal operational probabilistic theory. We named the two conditions strong and weak compatibility because there exist weakly compatible tests that are not strongly compatible, e.~g.~any pair of reversible quantum channels. For the case of channels, strong compatibility coincides with the notion of channels compatibility already given in Refs.~\cite{Teiko-incompatibility-2016,Teiko-incompatibility-2017}. We have then shown that for the case of observables strong and weak compatibility coincide. Finally we proved that a necessary and sufficient condition for full-compatibility, namely all tests of the theory are weakly compatible, is full-information without disturbance, namely all the information can be extracted without disturbance. Moreover, both conditions imply that all the systems of the theory are classical systems.

On the basis of our results we understand that generalizing the notion of measurement incompatibility to arbitrary devices, say instruments, is not straightforward. If for two measurements the ability of implementing them simultaneously (strong compatibility) or in a sequence (weak compatibility) is equivalent, for two instruments it turns out that it is not. This opens the problem of characterizing the two different kinds of incompatibility in terms of resources for quantum information tasks. Despite our results hold for arbitrary causal theories, the discrepancy between weak and strong compatibility already occurs for quantum channels, and a first analysis can focus on the quantum case. Recently, quantum measurement incompatibility, which is known to be a necessary condition for violating Bell inequalities and for steering, has been characterized in terms of quantum programmability of instruments, namely the temporal freedom of a user in issuing programs to a quantum device~\cite{PhysRevLett.124.120401}. It is reasonable to assume that one of the two notions of quantum instruments incompatibility should be the analogue, in terms of resources for processing of programmable instruments, to that of measurements.

Finally, as quantum channels are all weakly compatible, it is clear that the source of weak incompatibility for
quantum instruments resides in the irreversible process of information extraction, which cannot be reduced to a 
reversible interaction with the detector by any means. Indeed, if this were the case, the two notions of strong
and weak compatibility would coincide in the quantum case. On the other hand, it is pretty well understood that
strong incompatibility for quantum instruments is connected with the impossibility of cloning.


\acknowledgments We thank the financial support of Elvia and Federico Faggin Foundation (Silicon Valley Community Foundation Project ID\#2020-214365).

\bibliographystyle{apsrev4-2}
\bibliography{bibliography}

\end{document}

%% file: myQcircuit.tex
\usepackage[matrix,frame,arrow]{xy}
\usepackage{amsmath}
\newcommand{\qw}[1][-1]{\ar @{-} [0,#1]}

\newcommand{\gate}[1]{*{\xy *+<.6em>{#1};p\save+LU;+RU **\dir{-}\restore\save+RU;+RD **\dir{-}\restore\save+RD;+LD **\dir{-}\restore\POS+LD;+LU **\dir{-}\endxy} \qw}

\newcommand{\measureD}[1]{*{\xy*+=+<.5em>{\vphantom{\rule{0em}{.1em}#1}}*\cir{r_l};p\save*!R{#1} \restore\save+UC;+UC-<.5em,0em>*!R{\hphantom{#1}}+L **\dir{-} \restore\save+DC;+DC-<.5em,0em>*!R{\hphantom{#1}}+L **\dir{-} \restore\POS+UC-<.5em,0em>*!R{\hphantom{#1}}+L;+DC-<.5em,0em>*!R{\hphantom{#1}}+L **\dir{-} \endxy} \qw}

\newcommand{\multimeasureD}[2]{*+<1em,.9em>{\hphantom{#2}}\save[0,0].[#1,0];p\save !C *{#2},p+LU+<0em,0em>;+RU+<-.8em,0em> **\dir{-}\restore\save +LD;+LU **\dir{-}\restore\save +LD;+RD-<.8em,0em> **\dir{-} \restore\save +RD+<0em,.8em>;+RU-<0em,.8em> **\dir{-} \restore \POS !UR*!UR{\cir<.9em>{r_d}};!DR*!DR{\cir<.9em>{d_l}}\restore \qw}

\newcommand{\multigate}[2]{*+<1em,.9em>{\hphantom{#2}} \qw \POS[0,0].[#1,0];p !C *{#2},p \save+LU;+RU **\dir{-}\restore\save+RU;+RD **\dir{-}\restore\save+RD;+LD **\dir{-}\restore\save+LD;+LU **\dir{-}\restore}

\newcommand{\ghost}[1]{*+<1em,.9em>{\hphantom{#1}} \qw}
\newcommand{\Qcircuit}[1][0em]{\xymatrix @*=<#1>}

\newcommand{\pureghost}[1]{*+<1em,.9em>{\hphantom{#1}}}
\newcommand{\multiprepareC}[2]{*+<1em,.9em>{\hphantom{#2}}\save[0,0].[#1,0];p\save !C
  *{#2},p+RU+<0em,0em>;+LU+<+.8em,0em> **\dir{-}\restore\save +RD;+RU **\dir{-}\restore\save
  +RD;+LD+<.8em,0em> **\dir{-} \restore\save +LD+<0em,.8em>;+LU-<0em,.8em> **\dir{-} \restore \POS
  !UL*!UL{\cir<.9em>{u_r}};!DL*!DL{\cir<.9em>{l_u}}\restore}
\newcommand{\prepareC}[1]{*{\xy*+=+<.5em>{\vphantom{#1\rule{0em}{.1em}}}*\cir{l^r};p\save*!L{#1} \restore\save+UC;+UC+<.5em,0em>*!L{\hphantom{#1}}+R **\dir{-} \restore\save+DC;+DC+<.5em,0em>*!L{\hphantom{#1}}+R **\dir{-} \restore\POS+UC+<.5em,0em>*!L{\hphantom{#1}}+R;+DC+<.5em,0em>*!L{\hphantom{#1}}+R **\dir{-} \endxy}}
\newcommand{\poloFantasmaCn}[1]{{{}^{#1}_{\phantom{#1}}}}